\documentclass[a4paper,11pt]{article}
\usepackage{makeidx}
\usepackage[latin1]{inputenc}
\usepackage{amsfonts}
\usepackage{bm}
\usepackage{amssymb}
\usepackage{latexsym}
\usepackage{amsmath}
\usepackage{amscd}
\usepackage{amsthm}
\usepackage{rotating}
\usepackage{graphicx}
\usepackage{pifont}
\usepackage{curves}
\usepackage{fancyhdr}
\usepackage{epsfig}
\usepackage{pstricks}
\usepackage{pst-tree}
\usepackage{subfigure}
\usepackage{epic}
\usepackage{alltt}
\usepackage[all]{xy}
\usepackage{appendix}
\usepackage{tikz}
\usepackage{url}
\usepackage{color,soul}
\usepackage{enumitem}
\usepackage{mathtools}
\usepackage{mathalpha}
\usepackage{tabulary}
\usepackage{multirow}
\usepackage{graphics}
\usepackage{booktabs}

\usepackage{xcolor}
\usepackage[pdfencoding=auto, psdextra]{hyperref}
\hypersetup{
    colorlinks=true,
    citecolor=blue,
    breaklinks=true
    }

\usepackage{chngcntr}
\counterwithin{figure}{section}
\theoremstyle{plain}
\newtheorem{theorem}{Theorem}[section]
\newtheorem{prop}[theorem]{Proposition}

\newtheorem{corol}[theorem]{Corollary}

\newcommand{\be} {\begin{equation}}
\newcommand{\ee} {\end{equation}}
\newcommand{\bea} {\begin{eqnarray}}
\newcommand{\eea} {\end{eqnarray}}
\newcommand{\Bea} {\begin{eqnarray*}}
\newcommand{\Eea} {\end{eqnarray*}}

\DeclareMathOperator{\Trace}{Tr}

\def\P{\mathbb{P}}
\newcommand{\E}{\mathbb{E}}
\newcommand{\ind}{\mathbf{1}}

\newcommand{\N}{\mathbb{N}}
\newcommand{\rme}{\mathrm{e}}

\newcommand{\Laplace}{\mathfrak{L}}
\newcommand{\xx}{x^{k}_0}
\newcommand{\Qlg}{\mathbf{Q}_{\lambda, \gamma}}
\newcommand{\Lap}[3]{\Laplace_{#1, #2, #3}}
\newcommand{\LapG}[4]{\mathfrak{G}_{#1, #2, #3,#4}}
\newcommand{\hQ}{\widehat{\mathbf{Q}}}
\newcommand{\bfA}{\mathbf{A}}
\newcommand{\bfB}{\mathbf{B}}

\usepackage{xcolor}

\title{The epidemiological footprint of contact structures in models with two levels of mixing}
\author{Vincent Bansaye$^1$, Fran\c cois Deslandes$^2$, Madeleine Kubasch$^{1,2}$, Elisabeta Vergu$^2$}
\date{\small $^1$Ecole Polytechnique, Centre de math\'ematiques appliqu\'ees (CMAP), 91128 Palaiseau, France \\ $^2$MaIAGE, INRAE, Universit\'e Paris-Saclay, 78350 Jouy-en-Josas, France}

\begin{document}

\maketitle

\begin{abstract}

Models with several levels of mixing (households, workplaces), as well as various corresponding formulations for $R_0$, have been proposed in the literature. However, little attention has been paid to the impact of the distribution of the population size within social structures, effect that can help plan effective interventions.  We focus on the influence on the model outcomes of teleworking strategies, consisting in reshaping the distribution of workplace sizes. We consider a stochastic \emph{SIR} model with two levels of mixing, accounting for a uniformly mixing general population, each individual belonging also to a household and a workplace. The variance of the workplace size distribution appears to be a good proxy for the impact of this distribution on key outcomes of the epidemic, such as epidemic size and peak. In particular, our findings suggest that strategies where the proportion of individuals teleworking depends sublinearly on the size of the workplace outperform the strategy with linear dependence.  Besides, one drawback of the model with multiple levels of mixing is its complexity, raising interest in a reduced model. 
We propose a homogeneously mixing \emph{SIR} ODE-based model, whose infection rate is chosen as to observe the growth rate of the initial model.
This reduced model yields a generally satisfying approximation of the epidemic. These results, robust to various changes in model structure, are very promising from the perspective of implementing effective strategies based on social distancing of specific contacts. Furthermore, they contribute to the effort of building relevant approximations of individual based models at intermediate scales.

\end{abstract}

  
\noindent \textbf{Code availability.} \url{https://forgemia.inra.fr/francois.deslandes/communityepidemics}   \\

\noindent \textbf{Author Contributions} Vincent Bansaye and Elisabeta Vergu have supervised and designed the study. Fran\c cois Deslandes has been in charge of numerical aspects (implementation and development of simulation scenarios for numerical explorations). Madeleine Kubasch has developed theoretic aspects and contributed to the numerical part. All authors have contributed equally to the analysis and interpretation of the results. All authors have contributed equally to drafting the manuscript.

\section{Introduction}

The dynamics of an epidemic relies on the contacts between susceptible and infected individuals in the population. The number and characteristics of contacts has a major quantitative effect on the epidemic. In addition to the main features playing a role in the description of contacts, such as the age and propensity to travel of individuals \cite{daviesAgedependentEffectsTransmission2020, gilesDurationTravelImpacts2020}, the nature of the contact is also crucial: homogeneous mixing in closed structures (household, workplaces, schools,...) or related to other intermediate social structures (group of friends, neighbors...) e.g. \cite{houseDeterministicEpidemicModels2008}. The heterogeneity of contacts can be captured in network based models \cite{keelingNetworksEpidemicModels2005} or models with two levels of mixing \cite{ballGeneralModelStochastic2002}. These models distinguish a global level of mixing corresponding to a uniformly mixing general population, as well as a local level consisting in an overlapping groups model, meaning that each individual belongs to one or several small contact groups such as households and workplaces and schools. These  modeling frameworks or their simplified unstructured versions allow to tackle important questions related to the control of epidemic dynamics by acting specifically on these different population structures. Furthermore, the computation of the corresponding reproduction number, arguably one of the most important epidemic indicators, enables to assess control measures.

For the homogeneous mixing \emph{SIR} model, several important characteristics can be summarized by the reproduction number $R_0$. 
This threshold parameter indicates whether there may be a large epidemic outbreak, allows to calculate the final epidemic size and the fraction of the population that needs to be vaccinated in order to stop an outbreak, see e.g. \cite{heesterbeekConceptRoEpidemic1996, ballReproductionNumbersEpidemic2016} and references therein. It is also directly linked  to the exponential growth rate $r$ at the beginning of the epidemic, and has a clear interpretation as the mean number of individuals contaminated by a single infected individual in a large susceptible population. For models with two levels of mixing, however, the definition of a unique reproduction number combining these criteria has not been achieved yet. Instead, various reproduction numbers have been proposed, of which  \cite{ballReproductionNumbersEpidemic2016} have given an interesting overview. All of them respect the threshold of $1$ for large epidemic outbreaks, and they generalize one or another aspect of the traditional $R_0$.
Some of these reproduction numbers have the advantage of an intuitive interpretation. This is the case of the reproduction number $R_I$ introduced in the supplementary material of \cite{pellisThresholdParametersModel2009} for the household-workplace model, and which was previously introduced for household models \cite{beckerEffectHouseholdDistribution1995, ballEpidemicsTwoLevels1997}. Its definition relies on a multi-type branching process which focuses on primary cases within households and workplaces, grouping all secondary cases as descendants of the primary cases. Then $R_I$  is defined as the Perron root of the corresponding average offspring matrix. In this paper, we will show that this reproduction number has the advantage of being connected to further relevant information on the household-workplace epidemic, namely the proportions of infections occurring at each level of mixing. 
 
Nevertheless, a drawback of most of the reproduction numbers for household-workplace models that are described by \cite{ballReproductionNumbersEpidemic2016} is that by construction, they lose track of time. Indeed, in an effort to construct meaningful generations of infected individuals, the timing of the infections is neglected. As a consequence, contrary to the case of homogeneous mixing, there is no simple link between these reproduction numbers and the initial exponential growth rate. The only exception is $R_r$, a reproduction number which has originally been introduced by \cite{goldsteinReproductiveNumbersEpidemic2009} for household models, and whose definition has been extended by \cite{ballReproductionNumbersEpidemic2016} to household-workplace models. The definition of this reproduction number depends explicitly on the exponential growth rate $r$. But as far as we see,  it has no easy intuitive interpretation. It thus seems pertinent to complement the information yielded by a reproduction number such as $R_I$ with the  growth rate $r$. While simple closed analytic expressions seem out of reach, \cite{pellisEpidemicGrowthRate2011} have obtained an interesting characterization that we use and complement by more explicit expressions. 

Given the relative difficulty for computing reproduction numbers for models with several levels of mixing, \cite{goldsteinReproductiveNumbersEpidemic2009} have suggested, in the case of the simplest model with two levels of mixing, namely structured only in general population and households, to first estimate the growth rate from data, and to then compute $R_r$. \cite{trapmanInferringR0Emerging2016} have gone one step further, by proposing to first infer $r$ from data, and then totally neglect the population structure and approximate a reproduction number from $r$ using the formula linking the reproduction number to the exponential growth rate in the homogeneous mixing model. They find that this procedure is generally satisfactory, indicating that this procedure defines a homogeneous mixing model able to capture key aspects of the beginning of the epidemic. This makes one wonder to what extent it is possible in general to approach an epidemic spreading in a household-workplace model by a simple, unstructured, well parametrized compartmental model. Some work has been done in this direction by \cite{delvallerafoDiseaseDynamicsMean2021}. They have shown that it is possible to approach an \emph{SIRS} household-workplace model by a homogeneously mixing \emph{SIYRS} model, where $Y$ stands for infected but no longer infectious individuals, once the parameters have been well chosen. Hence, they obtain the first approximation of multi-level epidemic dynamics using homogeneous mixing compartmental ODEs.


Naturally, models with two levels of mixing raise the question of the way their social organisation characterized by small contact structures has an impact on major features of an epidemic. From the point of view of control, they constitute minimal models allowing to account for closures of workplaces or schools. For the past years, governments world-wide have implemented such non-pharmaceutical interventions (NPIs) in reaction to the COVID-19 epidemic. Since then, several studies have assessed the impact of these measures on the epidemic spread.
Both analysis of empirical studies \cite{mendez-britoSystematicReviewEmpirical2021} and simulation studies \cite{backhauszImpactSpatialSocial2022, simoySociallyStructuredModel2021} come to the conclusion that especially (partial) school closure and/or home working have a substantial impact on the epidemic. 
Together, these findings motivate an interest in mathematical models enabling a closer study of school and workplace closures, and more generally the effect of control measures targeting small contact structures.

In this paper, we are interested in the impact of the distribution of individuals in closed structures on epidemic dynamics. In order to address this question, we consider a stochastic \emph{SIR} model with two levels of mixing, namely a global and a local level. While the former corresponds to the general population, the latter is subdivided into two layers representing households and workplaces, respectively. Note that while our model does not explicitly distinguish schools, they can be considered as workplaces. In particular, we are motivated by and study the impact of control policies based on  differentiated social distancing. For some structures, in particular for households, it is natural to assume that their size distribution is fixed and control policies cannot act on it. For others, such as workplaces and schools, control measures aiming at contact reduction can be considered, COVID-19 epidemic having raised this issue in new manners. Focusing on workplaces, we study here how control strategies
which consist in modifying the structures' size distribution, can impact 
different epidemic outcomes.
More generally, we demonstrate through simulations that the size distribution of closed structures has a significant effect on epidemic dynamics, as assessed by the total number of infections and by  the initial growth rate of infection and by the maximal number of infected individuals along time. In particular, when both the number of individuals and structures are fixed, implying that the average structure size is constant as well, we show that these epidemic outcomes are sensitive to the variance of the structure size distribution. In short, balancing structure sizes reduces the impact of the epidemic.

One drawback of the model with two levels of mixing is that numerical simulations rely on good knowledge of several epidemic parameters, such as the rates of infection within each level, which may not be easy to assess. However, considering the significant impact of structure size distributions on epidemic outcomes and the fact that control measures may actively impact these distributions, it seems crucial not to neglect this particular population structure. This motivates the development of reduced epidemic models, which aim to be more parsimonious, while still being able to capture the impact of small structures on the epidemic thanks to a pertinent choice of parameters. 
Here, we propose such a reduced model, that we evaluate using simulations. 
It consists of a deterministic, homogeneously mixing \emph{SIR} model, whose infectious contact rate is chosen as to ensure that the reduced model and model with two levels of mixing share the same exponential growth rate.
Hence, we will see that the initial growth rate is the key parameter for reducing the full epidemic process at the macroscopic level.

The questions we consider here involve  quantities which capture some main features of the epidemics
which are relevant for specific phases of an epidemic.
Indeed, starting from a single infectious individual in a large population of size $N$, epidemic dynamics can be decomposed into three phases. This 
has been proven for simple models such as the homogeneously mixing \emph{SIR} model, for which we will detail these phases below.
However, this decomposition still holds in more complex models, including the model with two levels of mixing studied here. \\
\emph{Phase 1: random behavior in small population.} When the number of infected individuals $I(t)$ is of order $1$
, $I(t)$ is approximated by a linear birth and death process.  This approximation holds on finite time intervals, but also up to a time $T_N$ which tends to infinity when $N$ tends to infinity. More precisely, both processes coincide as long as the number of infected individuals is below $\sqrt{N}$ \cite{ballStrongApproximationsEpidemic1995}. Let us also mention \cite{barbourApproximatingReedFrost2004}, for comparison results until the infected population reaches sizes of order of $N^{2/3}$, for a discrete time counterpart of the \emph{SIR} model. \\
\emph{Phase 2: deterministic evolution and linear behavior.} When $1\ll I(s), I(t) \ll N$, the number of infected follows a deterministic and  exponential dynamic:  $I(t)\approx I(s)\rme^{(\beta-\gamma)(t-s)}$, where $\beta$ is the transmission rate and $\gamma$ the recovery rate. This approximation is valid as soon as $s, t$ tend to infinity but remain far  from the time $\log(N)/(\beta-\gamma)$, which corresponds to the entry in the macroscopic level. This deterministic phase allows to capture the initial growth rate of infection, $\beta-\gamma$, by considering the slope of the growth of $I$ on a logarithmic scale. We refer to \cite{bansayeSharpApproximationHitting2023}  and references therein for more precise results. \\
\emph{Phase 3: macroscopic deterministic behavior.} When the number of infected individuals is of order $N$,
the proportion of susceptible, infected and recovered individuals can be approximated by a macroscopic deterministic system. More precisely, letting $N$ go to infinity, the trajectories of $(S/N, I/N, R/N)$ converge in law on finite time intervals to the solutions of the \emph{SIR} dynamical system. The approximation is valid for any $t$ greater than $\log(N)/(\beta-\gamma)$. For accurate results, we refer in particular to \cite{barbourApproximatingEpidemicCurve2013} and \cite{barbourStochasticModelTransmission1978}. Let us also mention that fluctuations around the deterministic curve are of order $1/\sqrt{N}$ by classical Gaussian approximation (Chapter 7, Sections 4 and 5, in \cite{ethierMarkovProcessesCharacterization1986}).\\
In our study, \emph{phase 1} corresponds to the regime where stochasticity of the individual-based version of the \emph{SIR} model is observed in simulations. \emph{Phase 2} is the relevant regime for the definition of $R_I$ and the initial growth rate $r$. \emph{Phase 3} yields the deterministic macroscopic approximation, where stochasticity vanishes. It starts at a random time necessary to reach a macroscopic proportion of infected. This time represents the starting point of the comparison between the stochastic structured model and its reduced ODE-based counterpart we propose.

This paper is structured as follows. Section \ref{MM} presents the main modeling ingredients, such as a detailed description of the model with two levels of mixing and proper introduction of considered key parameters, as well as numerical settings for simulations. Section \ref{impactsizedistrib} is devoted to the study of the impact of the structure size distribution on some main epidemic outcomes, namely the reproduction number, the exponential growth rate, the peak size and the final epidemic size. For this purpose, two slightly different situations are considered. While the size of the population is always considered fixed, we first keep the total number of workplaces constant as well but vary the way individuals are distributed among these given workplaces, see Section \ref{impactnormal}. Second, we consider teleworking strategies, which differ from the previous setting as for simulations, these strategies amount to creating a new workplace of size one for each teleworking employee, see Section \ref{impacttws}. 
Finally, in Section \ref{reductionODE}, we propose an ODE reduction of the initial multi-level model based on the computation of the initial growth rate and assess its robustness. The paper concludes with a Discussion (Section  \ref{discussion}) of the main results on the impact of structure size distributions on epidemic dynamics, their robustness to different modeling assumptions, and their implications for control measures.

\section{Model with two levels of mixing: description, simulation approach, key parameters, simulation scenarios}
\label{MM}

\subsection{General model description}\label{model}

We consider an \emph{SIR}-type model with two levels of mixing by considering global and two types of local contacts following two local partitions of the population, see \cite{ballGeneralModelStochastic2002}. In addition to homogeneous mixing in the general population, contacts occur in households and workplaces of various sizes, in which the population is structured. Each individual belongs both to a household and a workplace, which are chosen independently from one another. Generally speaking, infection spreads through contacts between susceptible and infected individuals within each level of mixing, which are characterized by different contact rates among individuals as will be detailed below. Infected individuals recover at rate $\gamma$.

We distinguish two slightly different types of parametrization concerning contact description. For closed structures such as households and workplaces, we will use \emph{one-to-one} infectious contact rates  $\lambda_H$ and $\lambda_W$, respectively. In other words, within a household, if there are $s$ susceptibles and $i$ infected individuals, each susceptible is infected at rate $\lambda_H i$ (resp. $\lambda_W i$ for workplaces). This has the disadvantage to make the average number of contacts established by each individual grow with the size of the structure. This is tractable for structures of finite size, and a good enough approximation of contacts within very small structures, but it is not realistic at the scale of the general population. Instead, within the general population, when there are $s$ susceptibles and $i$ infectious individuals, each susceptible individual becomes infected at rate $\beta_G i/(N-1)$ where $N$ is the population size. Here, the parameter $\beta_G$ represents the \emph{one-to-all} infectious contact rate, which is the global rate at which an infected individual makes contact with all other individuals in the population. Hence the corresponding one-to-one infectious contact rate $\lambda_G^{(N)}= \beta_G /(N-1)$ is small. This allows to scale the contact rates when $N$ tends towards infinity, so that the mean number of contacts made by an infected individual remains constant. The global rate of infection in the population is then $\beta_G I S/(N-1)=\lambda_G^{(N)} I S$ where $S$ (resp. $I$) is the number of susceptible (resp. infected) individuals, and $S$ is indeed of order $N$. 

\subsection{Structure size distributions}
\label{structure-distributions}

Let us introduce the size distribution of households and workplaces, called $\pi^H$ and $\pi^W$, respectively.
When the number of structures is large, $\pi_k^H$ (resp. $\pi_k^W$) is the proportion of households (resp. workplaces) of size $k\geq 1$. 
The total number of individuals is $N$, which is fixed. Besides, all individuals belong to one (and only one) household and workplace, the latter being of size one for teleworking employees. Notice that  the following equivalences hold a.s.
$$N \, \sim \, N_H\sum_{k\geq 1} k \pi_k^H  \, \sim \,  N_W\sum_{k\geq 1} k \pi_k^W \qquad (N\rightarrow \infty),$$
where $N_H$ (resp. $N_W$) is the total number of
households (resp. workplaces). We define $m_H = \sum_{k\geq 1} k \pi^{H}_k$ (resp. $m_W = \sum_{k\geq 1} k \pi^{W}_k$) the average household (resp. workplace) size.


\subsection{Numerical simulation scenarios: structure size distributions and epidemiological parameters} 
\label{sec-param}

In the numerical explorations of the impact of the structure size distribution on the epidemic dynamics, we use the household size distribution observed in France in 2018 as reference distribution and also more generally, unless stated otherwise. We also provide a workplace distribution derived from the workplace size distribution of Ile-de-France in 2018, later called reference workplace size distribution, and we refer to Appendix \ref{appdx:sizedist} for detail. In particular, we assume homogeneous mixing within structures, which is unrealistic for large workplace sizes, and we thus have limited workplaces to size 50 at most. The household reference distribution is stated in Table \ref{tab:refhouseholds}, while the workplace reference distribution is shown in Figure \ref{fig:dist-insee}.

\begin{table}[ht]
\begin{center}
\begin{minipage}{\textwidth}
\caption{Reference household size distribution corresponding to the size distribution of households in France in 2018.}
\label{tab:refhouseholds}
\begin{tabular}{lcccccc}
  \toprule 
  Household size & 1 & 2 & 3 & 4 & 5 & 6 \\
  \midrule
  Proportion &  0.367 & 0.326 & 0.136 & 0.114 & 0.041 & 0.016 \\
  \bottomrule
\end{tabular}
\end{minipage}
\end{center}
\end{table}

\begin{figure}[!ht]
  \centering
\includegraphics[width=0.7\textwidth]{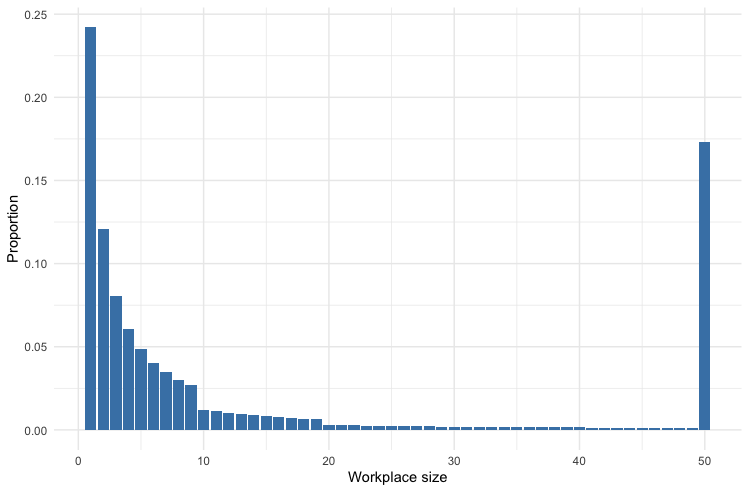} 
\caption{Reference workplace size distribution distribution derived from workplace size distributions in Ile-de-France.}
\label{fig:dist-insee}
\end{figure}


To study the impact of the average workplace size and workplace variance we provide the following sets of exploratory workplace size distributions: (A) a set of 160 workplace distributions with mean ranging from 3 to 30, different variances and maximal size 50; (B) a set of 100 workplace distributions with mean 20, different variances and maximal size 50; (C) a set of 100 workplace distributions with mean 7, different variances and maximal size 50. These workplace size distributions were generated using random mixtures of distributions, as explained in Appendix \ref{appdx:sizedist}. The sets of structure size distributions are summarized in Table \ref{tab:dist-bis}. 

\begin{sidewaystable}
\begin{center}
\begin{minipage}{\textheight}
\caption{Household and workplace size distribution sets.}
\label{tab:dist-bis}
\begin{tabular*}{\textheight}{@{\extracolsep{\fill}}lcccccc@{\extracolsep{\fill}}}
\toprule
\multicolumn{2}{c}{Name} & Number of distributions & Average size\footnotemark[1] & Variance\footnotemark[1] & Maximum size \\
\midrule
\multicolumn{2}{c}{Reference household size distribution}    &  1                       &    2.2          &    1.6      & 6           \\  
\multicolumn{2}{c}{Reference workplace size distribution}    &  1                       &    13.8          &    331.5      & 50        \\     
\midrule
\multirow{3}{5.5cm}{Exploratory workplace size distributions} & A    &  160                       &    3-30          &    9.8-554.3     & 50     \\
 & B    &  100                       &    20          &    42.4-544.0      & 50        \\     
 & C    &  100                       &    7          &    5.8-248.8      & 50        \\
\bottomrule
\end{tabular*}
\footnotetext[1]{Value or range of values.}
\end{minipage}
\end{center}
\end{sidewaystable}

Numerical exploration of the model was performed using a combination of various epidemic parameters.
We designed scenarios to cover a range of interesting situations that illustrate the mathematical properties of the model and its approximations as well as the expected epidemic behaviour of the model for several infectious diseases.
Our study covers several values of the reproduction number, from threshold values to values observed in real world epidemic such as the flu (\cite{ajelliRoleDifferentSocial2014}) or COVID-19 (\cite{locatelliEstimatingBasicReproduction2021,  galmicheExposuresAssociatedSARSCoV22021}). We also explore higher values of the reproduction number, closer to what is observed in highly contagious diseases such as chickenpox (\cite{silholModellingEffectsPopulation2011}).
Similarly, we cover a range of distributions of infections between structures and global mixing. We study balanced scenarios where the proportion of infection in global mixing is between 30 and 40\%, such as those observed in influenza or COVID-19 epidemics, as well as more contrasted situations where infections at the global or local level strongly dominate, as for chickenpox. An overview of all considered scenarios is given in Table \ref{tab:scenfeat}. 

\begin{table}[ht]
\begin{center}
\begin{minipage}{\textwidth}
\caption{Main features of the epidemic scenarios considered.}
\label{tab:scenfeat}
\begin{tabular}{llll}
  \toprule 
  Scenario & Reproduction number & Distribution of infections & Comment \\
  \midrule
  1 &  COVID-19 & balanced & COVID-19-like scenario \\
  2 &  high & balanced & \\
  3 &  threshold & balanced & \\
  4 &  high & mostly global mixing & \\
  5 &  high & mostly workplaces & \\
  6 &  flu & balanced & \\
  7 &  flu & mostly global mixing & \\
  8 &  flu & balanced & flu-like scenario\\
  9 & threshold & mostly structures & \\
  10 & threshold & balanced & \\
  11 & threshold & fully structures & \\
  \bottomrule
\end{tabular}
\end{minipage}
\end{center}
\end{table}

 More detail is given in Appendix \ref{annex_numer}, including values of reproduction number, epidemic growth rate and proportions of infection per layers for each scenario as defined in upcoming Section \ref{R0}. They correspond to the values reached considering the reference size distribution. The corresponding epidemic parameters are also provided. Without loss of generality, the recovery rate is set to 1 in all simulations. Illustration of the simulated final size for each scenario and exploratory workplace size distribution in the simulation study are provided in Figure S1 of the Online Resource.

\subsection{Simulations of the population structure and epidemic process}
\label{simulations-of-the-epidemic-process}

The epidemic process was simulated using the Gillespie algorithm (Stochastic Simulation Algorithm). A population structure (contacts between individuals) is first generated from the size distributions of households and workplaces, generated as described in Section \ref{sec-param} and Appendix \ref{appdx:sizedist}. Individuals are listed from $1$ to $N$ and  placed in structures randomly by applying the following iterative process: (i) for each type of structure, randomly  select a structure size $k$ with probability given by the  size distribution; (ii) if the number $n$ of individuals that have not yet been assigned to a structure of the same type exceeds $k$, randomly select $k$ individuals among those; otherwise, group all remaining individuals in a structure of size $n < k$.
For a given population structure, the algorithm computes the rates of events, \textit{i.e.} infection events in households, workplaces and the general population, and recovery, as described in Section \ref{model}.
These rates are used to derive the next event, \emph{i.e.} infection of a susceptible individual or recovery of an infected individual, and the corresponding event time. The epidemic is initiated with a single infected individual, selected uniformly at random in the population. For each run of the epidemic process, we compute several classical summary statistics: (i) the final epidemic size, \textit{i.e.} the number of individuals that are in the recovered state when the number of infected individuals becomes 0; (ii) the infectious peak size, \textit{i.e.} the maximum number of infectious individuals occurring simultaneously over the course of the epidemic; (iii) the infectious peak time, \textit{i.e.} the time at which the infectious peak occurs.

\subsection{Simulation of teleworking strategies}\label{teleworking}
\label{sec:meth-telework}

\begin{figure}
    \centering
    \includegraphics[width=\textwidth]{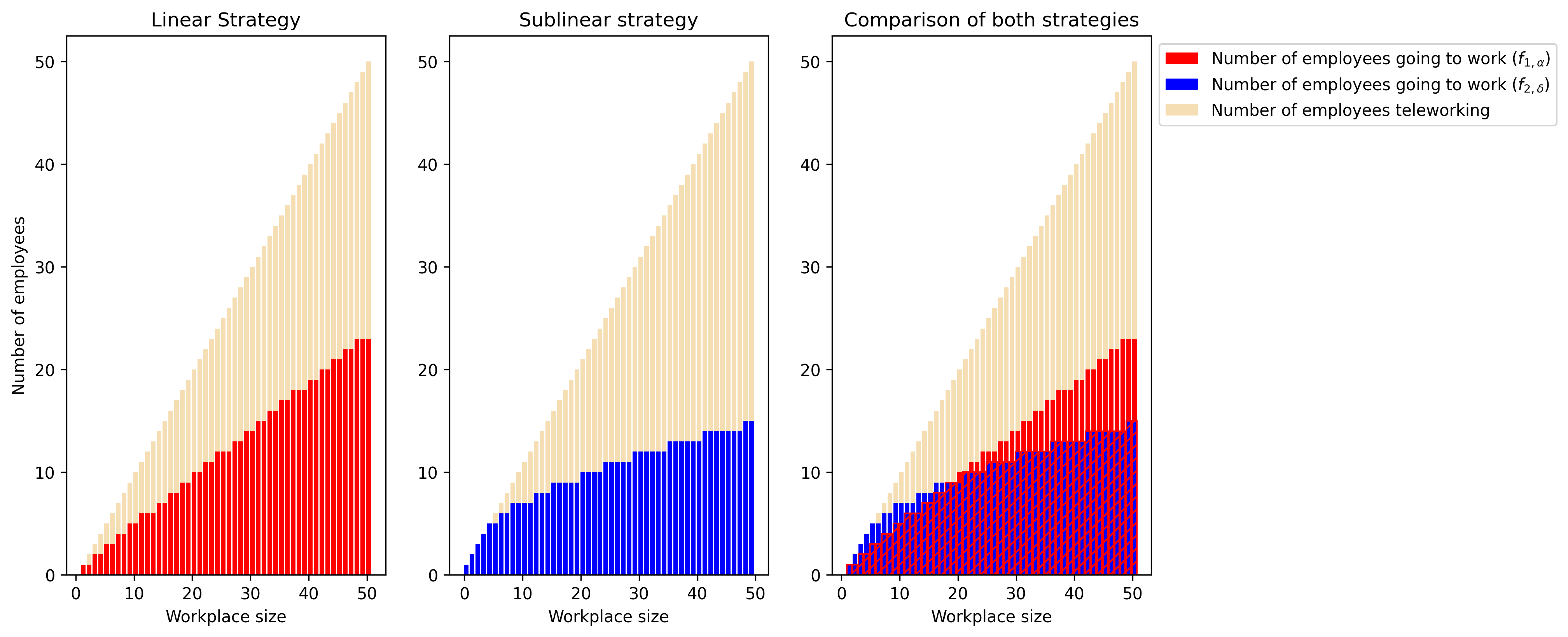}
    \caption{For each workplace size, number of employees coming to work on site or teleworking according to either the linear or the sublinear strategy. The parameters $\alpha$ and $\delta$ are chosen in order to observe an average proportion of employees teleworking per workplace equal to 0.5 for the uniform workplace size distribution $(\alpha = 0.46, \delta = 2.01)$.}
    \label{fig:twstrats}
\end{figure}

We evaluated the evolution of different epidemic outcomes for varying proportions of teleworkers using two strategies, as illustrated in Figure \ref{fig:twstrats}: \\
(i)  Linear strategy. The workplace size distribution is modified according to the function $f_{1,\alpha}(k)= \lceil \alpha k \rceil$ for $\alpha \in [0,1]$, where $\lceil  \cdot \rceil$ is the usual ceiling function. This means that a work place of size $k$ becomes a workplace of size $\lceil \alpha k \rceil$ and the remaining $k-\lceil \alpha k \rceil$ individuals now telework. (ii) Sublinear strategy.  The workplace size distribution is modified according to the function $f_{2,\delta}(k)=\lceil \delta k^{\frac{1}{2}} \rceil\wedge k$, where $\delta\geq 0$.This means that a work place of size $k$ becomes a workplace of size $\lceil \delta k^{\frac{1}{2}}\rceil$ and the remaining $k-\lceil \delta k^{\frac{1}{2}}\rceil$ individuals now telework.

The rationale behind the sublinear strategy is that withdrawing an individual from a large structure has a stronger impact in terms of number of contacts. Notice that we consider here the exponent $1/2$ for the sublinear strategy, but one could more generally consider any exponent $q < 1$.

\section{The impact of the size distribution of closed structures and assessment of teleworking strategies}
\label{impactsizedistrib}

\subsection{Outbreak criterion, $R_I$ and type of infection}
\label{R0}

Various notions of reproduction numbers have been proposed, as a compromise between complexity of the computations and epidemiological interpretation. The idea is to capture the mean number of infections caused more or less directly by 
one single "typical" individual. This concept is primarily defined in the first steps of the epidemic, which usually can be approximated by a branching process, whose mean satisfies a linear ODE. 
A typical individual corresponds then to a uniform sample in the corresponding population. The reproduction number is  delicate to define for epidemiological processes with multi-level contacts, such as the one we consider here. 
We recall that each infected individual infects an individual outside his structure with rate $\beta_G S/(N-1)$, which is in the \textit{phase $1$} approximated by $\beta_G$, since the number of susceptibles $S$ is approximately  $N$. As a consequence, the mean number of individuals directly infected in the general population by a single infected individual is $\beta_G/\gamma$.

Following the Supplementary Material of \cite{pellisThresholdParametersModel2009}, we structure the infected population following the origin of the infection and consider successive generations of infected individuals: 
\begin{equation*}
    (I_n^G, I_n^H,  I_n^W)_{n\geq 0}, \qquad  I_n=I_n^G+I_n^H+I_n^W.
\end{equation*}

\noindent Processes $I_n^G$, resp. $I_n^H$ and $I_n^W$, count the number of individuals in generation $n$, which have been infected through the mean field, respectively in the household and in the workplace. Hence $I_n$ is the total number of infected  individuals in generation $n$.
At time $0$, we assume $I_0=I_0^G=1$. The next generation $n+1$ of infected individuals is created by considering the number of direct infections $I_{n+1}^G$ in the general population, plus the local epidemic triggered within structures. This process is illustrated in Figure \ref{fig-branching-model}.

To compute the mean number of infections per generation, it is necessary to compute the mean number of individuals infected during the epidemic triggered by a single infected individual in a given structure. 
Thus, we introduce  $i_H(k)$ (resp. $i_W(k)$), the average total number of infections starting from one infected individual in a closed population of size $k$, one-to-one contact rate $\lambda_H$  (resp. $\lambda_W$) and recovery rate $\gamma$ as defined in Section \ref{model}. It corresponds to the number of infections caused by a single infected individual which introduces the epidemic into his household (resp. his workplace) of size $k$.
Recalling that $m_H = \sum_{k\geq 1}  k \pi^{H}_k$ (resp. $m_W = \sum_{k\geq 1} k \pi^{W}_k$), we define  
\begin{equation*}
   \widehat{\pi}_k^H= \frac{k \pi_k^H}{m_H}, \; \text{resp.} \; \widehat{\pi}_k^W= \frac{k \pi_k^W}{m_W},
\end{equation*}
as the size biased distribution
of structure sizes, which naturally defines the household (resp. workplace) size distribution of an individual chosen uniformly at random in the population. 
Then the numbers of infected individuals at each level triggered by an infected individual whose size structure is distributed according to the size biased law are defined by:
\begin{equation} \label{eq:nhat}
\mathcal I_G=\frac{ \beta_G}{\gamma}, \qquad \mathcal I_H=\sum_{k} \widehat{\pi}_k^H i_H(k), \qquad \mathcal I_W=\sum_{k} \widehat{\pi}_k^W i_W(k).
\end{equation}

\begin{figure}
	\includegraphics[width=\textwidth]{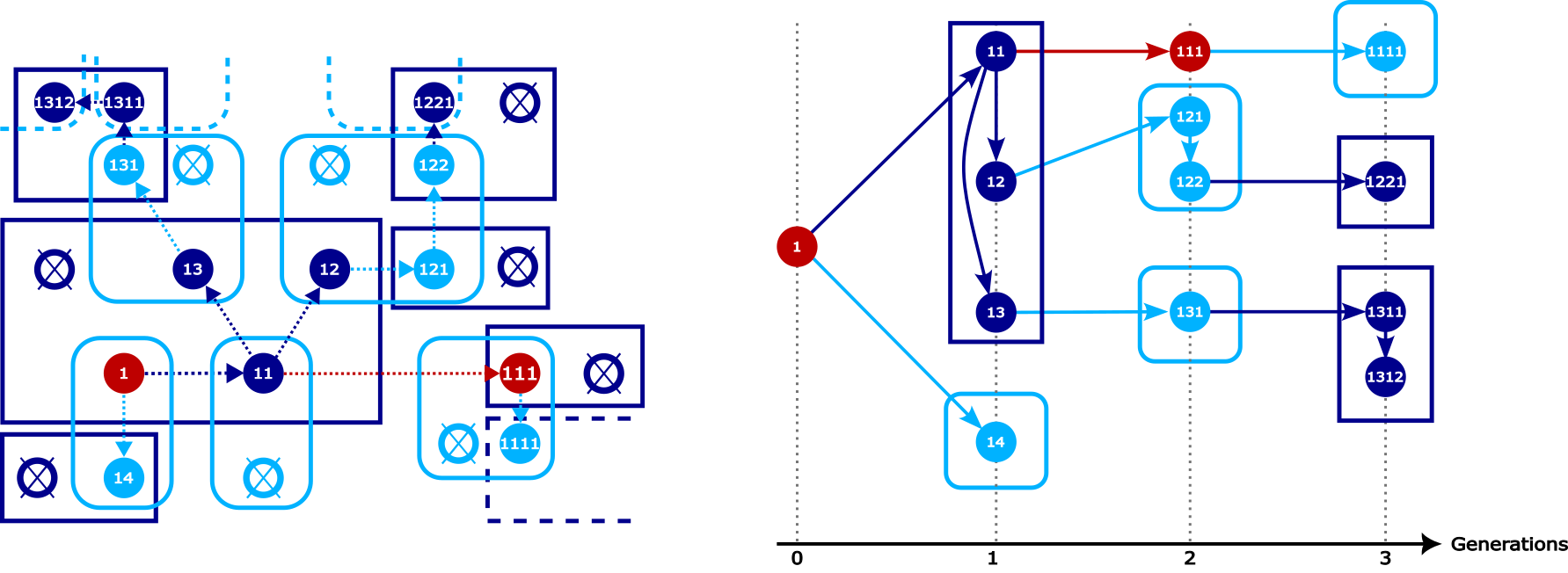}
	\caption{Example of an epidemic with two levels of mixing (in general population and within structures). The spread of the epidemic is shown on the left. Households and workplaces are delimited respectively in dark and light blue. Only structures containing infected individuals are shown. Individuals who have been infected during the epidemic appear as plain dots, whose colors indicate the means by which the infection occurred: through the general population (red), within households (dark blue) or within workplaces (light blue). The arrows keep track of the spread of the disease, pointing from the infector to the newly infected. Their color refers again to the type of infection. Members of a structure who have not been infected during the within-structure epidemic are represented as crossed circles. In the case of the branching model, they are never going to be infected, as a secondary introduction of the disease in an already infected structure is negligible  at the beginning of the epidemic. The different generations of infected as seen by the branching process are represented on the right. Colors still encode the way each infected is contaminated, and arrows represent the true order of infections as depicted on the left. 
	The branching genealogy is depicted through lexicographic labeling of individuals. Labels have been reported on the left panel, as to simplify identification of infected individuals in both means of representation.}
	\label{fig-branching-model}
\end{figure}

Following \cite{pellisThresholdParametersModel2009}, the expectation $\E((I_n^G, I_n^H,  I_n^W)^T)$ can be approximated by a sequence $X_n$
satisfying the following linear induction
$$X_{n+1}=AX_n$$
where $A$ is the mean reproduction matrix: 
\begin{equation}
\label{eq:reprod_matrix}
A=\begin{pmatrix} 
\mathcal I_G & \mathcal I_G & \mathcal I_G \\  \mathcal I_H &  0    & \mathcal I_H \\  
\mathcal I_W &  \mathcal I_W & 0  \end{pmatrix}.
\end{equation}
As $A$ is a primitive matrix, Perron Frobenius theorem yields the asymptotic behavior of $X_n=A^nX_0$ using its positive eigenelements, see \cite{athreyaBranchingProcesses1972}. More precisely, the unique  positive vector  $P=(p_G, p_H, p_W)^T$  solution of
\begin{equation} \label{eq-pgphpw}
AP=R_I P, \qquad p_G+p_H+p_W=1,
\end{equation}
gives the proportion of infected individuals from each source: general population, households, workplaces.
The associated positive eigenvalue is
$$R_I=\parallel AP\parallel_1$$
which corresponds to the mean reproduction number 
\begin{equation} \label{eq:R0}
R_I=\mathcal I_G +(1-p_H)\mathcal I_H+ (1-p_W)\mathcal I_W.
\end{equation}
When $R_I>1$, the process $(I_n)_{n \geq 1}$ survives with positive probability and on this event a.s. grows geometrically fast with speed $R_I$ yielding a supercritical regime, under an additional moment assumption on the number of infections, see also \cite{athreyaBranchingProcesses1972}. We observe that the vector $P$ gives  the origin of infections for large times.

As illustrated in Figure \ref{fig:R0_outbreak}, $R_I$ can play the role of an outbreak criterion. The final size of the epidemic is plotted against the value of $R_I$ for three parameter sets (scenarios 9, 10 and 11 in Table \ref{tab:scenfeat}). These epidemic parameter sets, combined with the set of workplace size distributions, allow to cover a large range of values for $p_G$, $p_H$ and $p_W$ (more precisely, $p_G$ between 0 and 0.5 and $p_W$ between 0.05 and 0.65).

\begin{figure}[!ht]
  \centering
\includegraphics[width=0.7\textwidth]{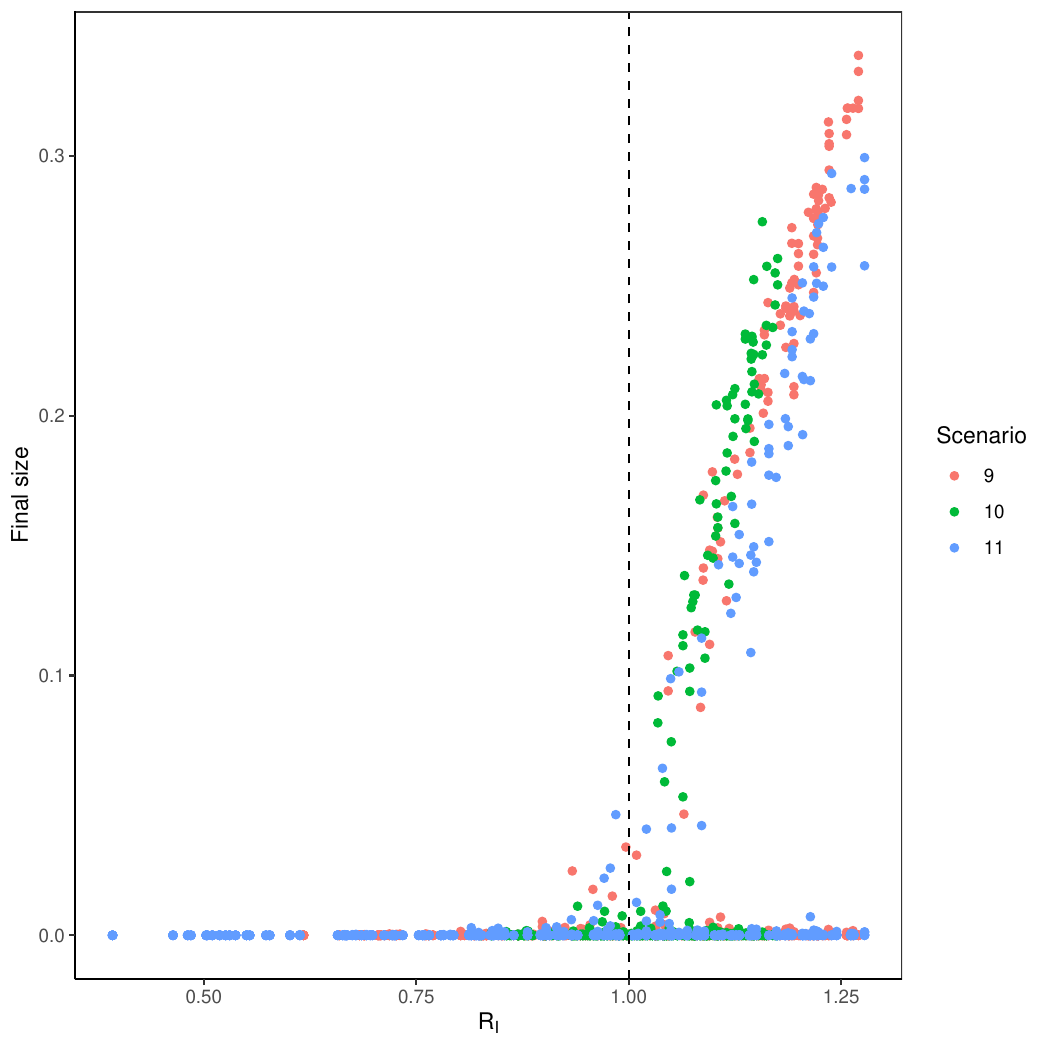} 
\caption{Simulated values of the final size of epidemics of the stochastic structured model as a function of $R_I$. Simulations are performed for the reference household distribution, exploratory workplace size distributions from set A in Table \ref{tab:dist-bis} and epidemiological scenarios 9, 10 and 11 from Table \ref{tab:scenfeat}. For each combination of workplace size distribution and epidemic scenario, 10 runs are performed.
}
\label{fig:R0_outbreak}
\end{figure}

Let us end this section with a brief description of the numerical computation of $R_I$. For a given set of structure size, transmission parameter and recovery rate, we simulate the within structure epidemic using Gillespie's algorithm, and record the final epidemic size. We thus calculate the average size of epidemics in isolated structures from simulations of this epidemic process (default value of the number of runs: 10000, and 50000 for larger workplaces). Notice that the average final epidemic size  could also be obtained through analytic results, see for instance Section 6.4 in \cite{baileyMathematicalTheoryInfectious1975}. The value of $R_I$ is then obtained from Equation \eqref{eq:R0}, as the largest  eigenvalue of the matrix defined in \eqref{eq:reprod_matrix}, by replacing unknown quantities by simulated quantities. The proportions of infections occurring within each layer of mixing are obtained from the associated eigenvector.

\subsection{The effect of structure size distribution on epidemic outcomes}
\label{impactnormal}

\begin{figure}[!ht]
  \centering
\includegraphics[width=\textwidth]{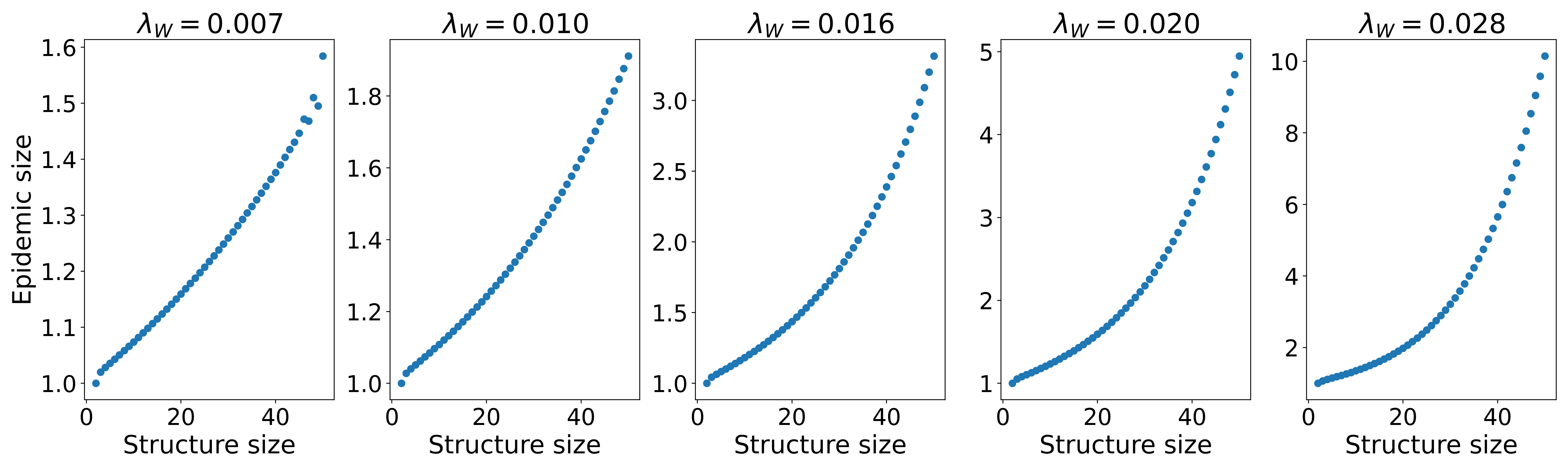} 
\caption{
Average final epidemic size within workplaces, i.e. for the uniformly mixing SIR model with one-to-one contact rate $\lambda_W$ and removal rate $\gamma=1$, as a function of workplace size. The average epidemic final size is computed using explicit formulas obtained from the Sellke construction, see for instance Section 2.4 in \cite{anderssonStochasticEpidemicModels2000}. Different values of one-to-one contact rate $\lambda_W$ have been considered, in order to cover a representative sample of scenarios 1 to 11 of Table B1, for the reference structure size distributions. More precisely, they correspond from left to right to scenarios 10, 4, 11, 2 and 5.
}
\label{fig:final-size}
\end{figure}

In this section, we will consider both the number of individuals and the number of workplaces as fixed. For the latter, one can imagine that this is due to logistic constraints, as there are only a certain number of offices (or classrooms) to dispose of. In other words, we are interested in understanding the epidemic impact of the way employees are assigned to those given workplaces. 

From the previous section, one may notice that the workplace size distribution has a direct impact on $R_I$ through $\mathcal{I}_W$, \emph{i.e.} the average number of infections occurring within workplaces
under the size-biased workplace size distribution $\hat{\pi}^W$. 
In particular, it follows from Equation \eqref{eq:R0} that diminishing $\mathcal{I}_W$ is enough to ensure that $R_I$ decreases. 
Further, Figure \ref{fig:final-size} illustrates that for most epidemic scenarios considered (Tables \ref{tab:scenfeat} and \ref{tab:scenarios}), the average number of infected $i_W(k)$ caused by a within-workplace epidemic in a workplace of size $k$ can be reasonably approximated in some scenarios by a linear function of the workplace size $k$. We deduce that, up to a constant $c$,
\begin{equation}
    \mathcal{I}_W = \sum_{k \geq 1} \hat{\pi}_k^W i_W(k) \approx c\sum_{k \geq 1} \hat{\pi}_k^W k = \frac{c}{m_W} \sum_{k \geq 1} k^2 \pi_k^W = c\frac{m_W^{(2)}}{m_W},
\end{equation}
where $m_W^{(2)}$ designates the second moment of the workplace size distribution.

Since we suppose both the population size and the number of workplaces to be fixed, it follows that the average workplace size $m_W$ is constant as well. As thus, in order to reduce $\mathcal{I}_W$, it is enough to reduce $m_W^{(2)}$. At fixed expected workplace size, modifying $m_W^{(2)}$ is strictly equivalent to modifying the workplace size variance. Since the latter has a more direct and intuitive interpretation, we will focus on the variance of $\pi^W$ as a natural candidate for the epidemic impact of the workplace size distribution.

In order to assess this impact, we will proceed by numerical exploration. A variety of workplace size distributions of average fixed at 20 and different variances have been considered, corresponding to exploratory workplace size distributions set B of Table \ref{tab:dist-bis}. For each of these distributions and for epidemic scenarios 1, 2, 4 and 5 we have computed the epidemic growth rate as explained in Appendix \ref{supp:nuemrical_methods}, before evaluating through simulations the epidemic size and the peak size. Results have been reported in Figure \ref{fig:Variance} (and Figure S2 of the Online Resource for additional scenarios), which thus illustrates the impact of the variance of the workplace size distributions on our selected epidemic outcomes (growth rate, final size, and peak size).
This figure shows that the workplace size variance has a linear impact on these epidemic outcomes, observed for various values of average workplace size, see also Figures S3 and S4 of the Online Resource.
Thus, the variance appears as a relevant indicator of the epidemic impact of the workplace size distribution. This also is of interest for the design of efficient control policies such as teleworking and (partial) closure of schools, as will be explored in the next section. 

\begin{figure}[!ht]
  \centering
\includegraphics[width=0.7\textwidth]{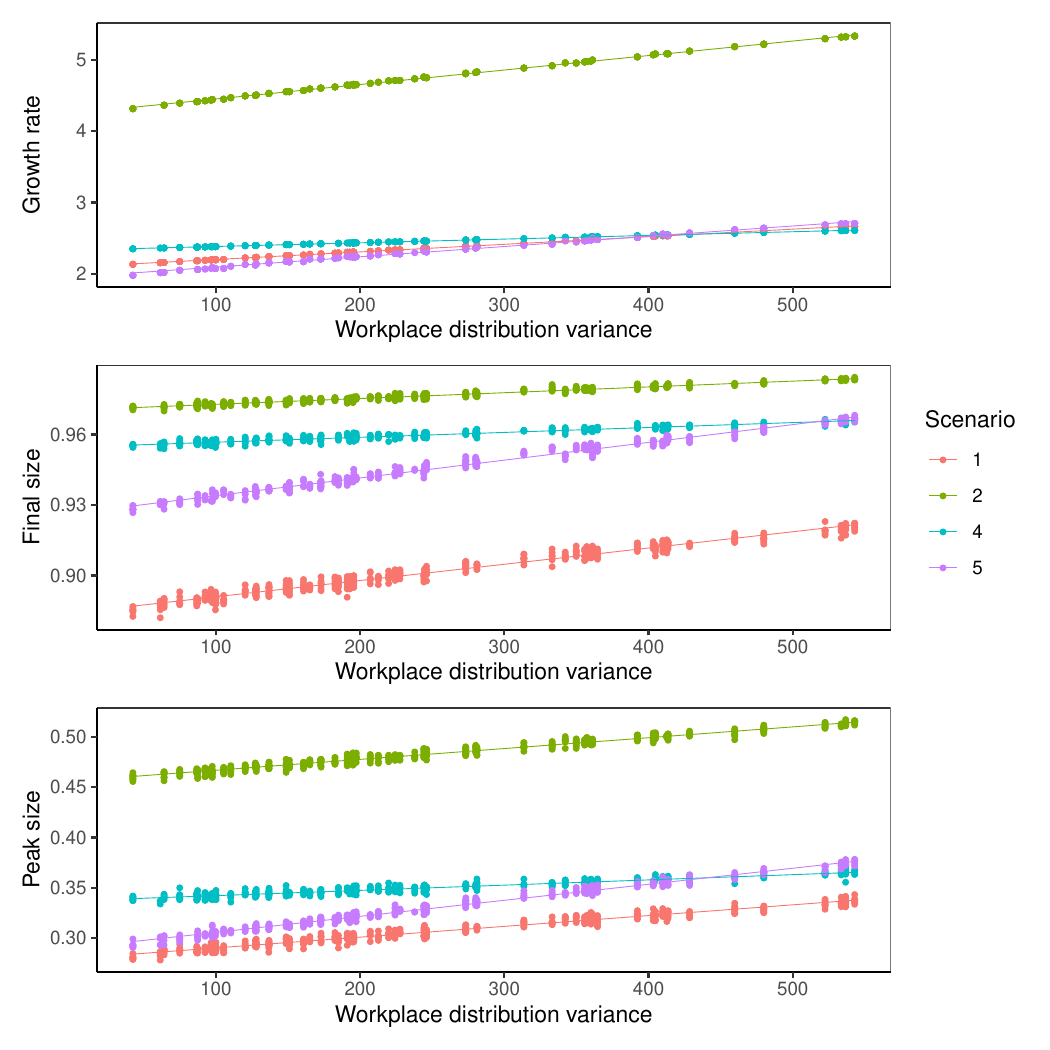} 
\caption{Influence of the variance of the workplace size distribution on the growth rate (top), epidemic final size (middle) and epidemic peak size (pbottom). Simulations of the stochastic structured model were performed with the reference household size distribution, exploratory workplace size distribution set B with average workplace size of 20 from Table \ref{tab:dist-bis} and epidemic scenario 1, 2, 4 and 5 from Table \ref{tab:scenfeat}. Simulations were repeated 10 times for each combination of scenario and workplace size distribution.
}
\label{fig:Variance}
\end{figure}

\subsection{Teleworking strategies}
\label{impacttws}

Teleworking, a strategy to mitigate disease outbreaks, results in changes in the distribution of workplace sizes. These changes have an impact on the value of $R_I$, which has been shown to be a threshold criterion for epidemics, and more generally on the different epidemic outcomes.

Two teleworking strategies, formalized in Section \ref{sec:meth-telework}, were assessed: (i) a linear strategy, where  the same proportion of teleworking is applied equally to all workplaces, and (ii) a sublinear strategy where teleworking is more prevalent in larger workplaces. The motivation for such a strategy is twofold: larger workplaces allow for larger within-workplace epidemics, and they are expected to be better equipped to mitigate the economic impact of teleworking on the firm. We will indeed show that such a strategy has a beneficial health outcome.

Figure \ref{fig:teleworking} illustrates the behaviour of the two teleworking strategies as a function of the teleworking rate, which is defined as the proportion of individuals in the population that do not have contacts in a workplace. Implementation of the teleworking strategies consists in adjusting parameters $\alpha$ and $\delta$ from Section \ref{sec:meth-telework} to obtain a prescribed value of the teleworking rate. The proportions of infections in the different structures are the same for both strategies at the threshold $R_I=1$. The findings show a large reduction in the proportion of infections occurring in the workplaces. The sublinear strategy reaches the threshold for a lower teleworking rate, which indicates that this strategy has a lower impact on workplace organisation for similar epidemic outcomes.

Figure \ref{fig:teleworking2} illustrates, for the reference workplace size distribution and epidemic scenario 1 of Table \ref{tab:scenfeat}, that even if the threshold cannot be reached by simply applying teleworking strategies, the sublinear strategy still outperforms the linear one. In particular, it shows that for the same global teleworking rate, the final size of the epidemic is lower for the sublinear strategy, except for the highest teleworking rates.

In other words, using sublinear teleworking policies (and more generally sublinear strategies for the closure of structures) allows either to reduce the need of teleworking in order to attain a given epidemiological outcome, or to reduce more strongly the epidemiological outcome for a given teleworking rate. Both effects may even be combined: more people go to work, but the epidemic outcome is reduced when compared to the linear strategy. 

\begin{figure}[!ht]
  \centering
\includegraphics[width=0.7\textwidth]{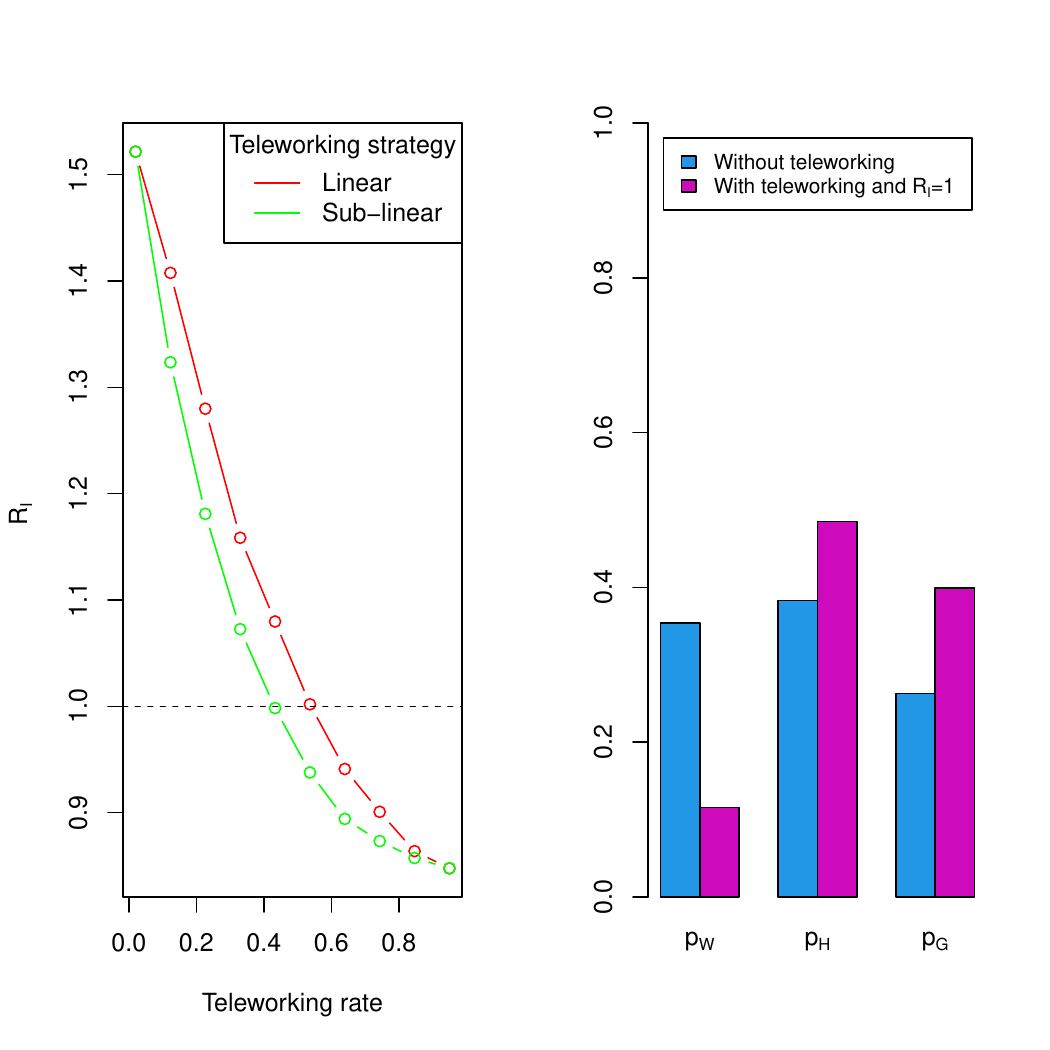}
\caption{Effect of linear and sublinear teleworking strategies, using scenario 6 from Table \ref{tab:scenfeat} and reference household and workplace distributions from Table \ref{tab:dist-bis}. Diminution of $R_I$ as a function of teleworking rate (left) and proportions of infection in the different structures (right).}
\label{fig:teleworking}
\end{figure}

\begin{figure}[!ht]
  \centering
\includegraphics[width=0.7\textwidth]{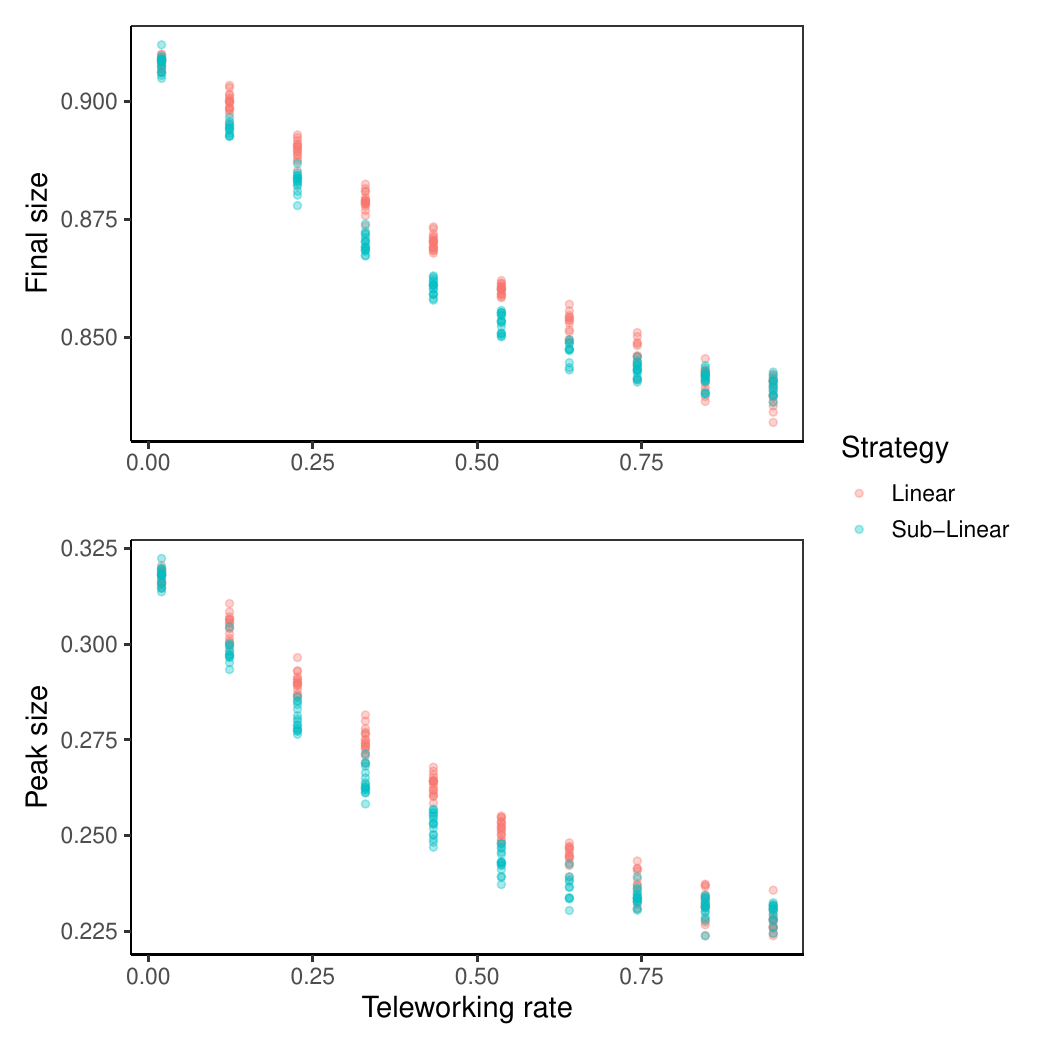} 
\caption{
Impact of linear and sublinear teleworking strategies on the final epidemic size (top) and epidemic peak size (bottom). Simulations of the stochastic structured model were performed using the reference household size distribution and workplace size distribution from Table \ref{tab:dist-bis}, with epidemic parameters from scenario 1 of Table \ref{tab:scenfeat}. 10 runs of were performed for each simulation scenario and distribution.
}
\label{fig:teleworking2}
\end{figure}

\subsection{Robustness to the form of the infection term}
\label{robustinteract}

Assuming linear growth of the number of infectious contacts per susceptible with the number of infected is often an overestimation. 
Thus, in this simulation study, we focus on having an infection rate within social structures growing sublinearly with the number of infected individuals in the structures. More precisely, we assume that within a household ($X=H$) or workplace ($X=W$) containing $I$ infected and $S$ susceptible individuals, the next infection occurs at rate $\lambda_X S \sqrt{I}$.

The observed linear effect of variance on the peak size and final size of the epidemic remains valid when the model is modified to use a sublinear infection rate in households and workplaces. Figure \ref{fig:Variance_sqrt1} shows the effect of the variance of the workplace size distribution with fixed mean on the epidemic outcomes for several scenarios (additional results for other scenarios can be found in Figure S5 of the Online Resource). This impact of the variance appears to hold for all scenarios, suggesting that the effect holds true regardless of both the epidemic speed and the proportions of infections that occur within structures. We also confirmed the validity of this result for workplace distributions with smaller mean sizes, as illustrated in Figures S6 and S7 of the Online Resource.

\begin{figure}[!ht]
  \centering
\includegraphics[width=0.7\textwidth]{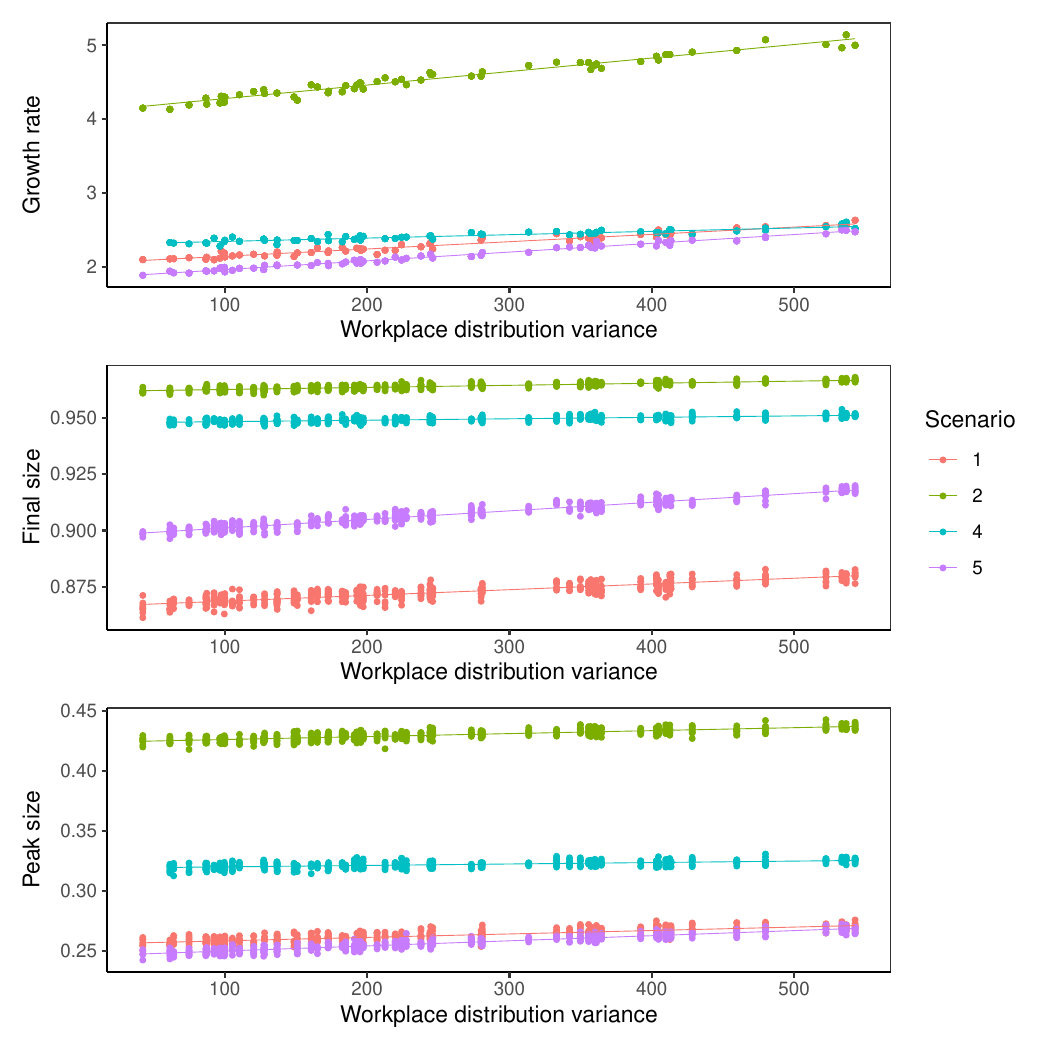} 
\caption{Influence of the variance of the workplace size distribution on the growth rate (top), epidemic final size (middle) and epidemic peak size (bottom). Simulations of the stochastic structured model with sublinear infection rates in households and workplaces were performed with the reference household size distribution, exploratory workplace size distribution set B with average workplace size of 20 from Table \ref{tab:dist-bis} and epidemic scenarios 1, 2, 4 and 5 from Table \ref{tab:scenfeat}. Simulations were repeated 10 times for each combination of scenario and workplace size distribution.
}
\label{fig:Variance_sqrt1}
\end{figure}

\section{Reduction to compartmental ODEs based on the initial growth rate}
\label{reductionODE}

In this section, our aim is to propose a relevant reduction of the multi-level contact process, when the total population is large $(N\gg 1)$ and the number of infected individuals too, corresponding to \textit{phase $2$} and \textit{phase $3$} presented in the introduction. We propose a deterministic reduction which keeps track of the multi-level structuring of contacts, but has a low dimension and depends on few parameters only. It thus allows to see the effect of structure size distributions and control policies modifying them at a low computational cost. We show that the key parameter to achieve this reduction is the initial growth rate. As expected, it captures the initial growth of the size of the infected population. Actually, simulations show that it also allows  for a  relevant prediction of the rest of the epidemic, see Section \ref{sec-reduction} for details on the interest and limitations of this reduction.

We assume that the total population size $N$ is large and consider an approximation in an infinite population. As for the branching approximation considered in Section \ref{impactsizedistrib}, we focus on the beginning of the epidemic (\textit{phase $1$} and \textit{phase $2$}).  As households and workplaces are chosen independently from one another and for each individual, this implies that whenever an infection occurs in the general population, it will almost surely affect an individual whose household and workplace are entirely susceptible otherwise. Similarly, an infection taking place within a household will cause an infection within an otherwise susceptible workplace, and vice-versa. Some time is needed to reach a large (but still negligible compared to $N$) number of infected individuals and forget the peculiar initial condition. Perron Frobenius theorem  allows to get a deterministic growth rate, which is observable in \textit{phase $2$}. More precisely, it is observed at the beginning of phase 2 and more generally before the infected population is too close to $N$. For more quantitative results on this point for the SIR model, we refer to Proposition 5.1
of  \cite{bansayeSharpApproximationHitting2023}. It is called the initial growth rate $r$. It will play a crucial role in reducing and analyzing the process in the deterministic phase with a macroscopic number of infected individuals (\textit{phase $3$}).

For the stochastic \emph{SIR} model in large homogeneously mixing populations, the initial growth rate can readily be obtained \cite[Section 1.2.3]{diekmannMathematicalEpidemiologyInfectious2000}. Let us briefly explain the heuristics of the reasoning. Consider $\mathbf{i}(t)$ the number of new infections in the population occurring at time $t$ after the start of the epidemic. It is easy to see that $\mathbf{i}(t)$ satisfies a renewal equation, which may be used to deduce an implicit equation for the exponential growth rate $\bar{r}$. Indeed, suppose that $\mathbf{i}(t) = C\rme^{\bar{r} t}$ for some constant $C$, and let $\zeta(\tau)$ denote the average rate at which an individual who has been infected $\tau$ units of time ago transmits the disease. Then $\bar{r}$ is characterized as follows:
\begin{equation}
\label{growth-rate-hSIR}
\Laplace(\zeta)(\bar{r}) \coloneqq \int_{0}^{\infty} \zeta(\tau) \rme^{-\bar{r} \tau} d\tau = 1,
\end{equation}
where $\Laplace$ designates the Laplace transform operator. In order to conclude, it remains to make $\zeta$ explicit, and at the beginning of the epidemic, one readily obtains the approximation $\zeta(\tau) \approx \beta \rme^{-\gamma \tau}$. Injecting this into the implicit equation $\Laplace(\zeta)(\bar{r}) = 1$ leads to the well-known growth rate $\bar{r} = \beta - \gamma$. This derivation can rigorously be obtained using branching approximations \cite{modeLifeCycleModels2000}.

Exponential growth of infections is also observed when household-workplace structures are added to homogeneous mixing, and \cite{pellisEpidemicGrowthRate2011} characterize the associated growth rate $r$. Similarly to what has been done for the reproduction number, they aggregate within-structure epidemics to facilitate the mathematical analysis of the model. This leads to a point of view where an infected household contaminates a new workplace each time an infection occurs during the within-household epidemic, and vice-versa. Using equation (\ref{growth-rate-hSIR}), this allows for the exact characterization of $r$ as the unique solution of an implicit equation, which can be solved numerically. This motivates the study of the Laplace transform
(\ref{growth-rate-hSIR}) of the average rate at which infections occur during the course of within-structure epidemics, which captures the dynamics of these infections. 

Here, we follow and complement the approach in \cite{pellisEpidemicGrowthRate2011}, mainly by providing more explicit expressions of the key quantities involved. The following Section \ref{sec-laplace} introduces our contribution, which lies in Proposition \ref{prop-hatq} and its corollaries. Subsequent Section \ref{sec-pellis} summarizes the work of \cite{pellisEpidemicGrowthRate2011}, and allows to position our contribution in the context of their work. 

\subsection{Laplace transform of the infection rate in a uniformly mixing population}
\label{sec-laplace}

The main point lies in understanding the dynamics of the stochastic \emph{SIR} model in a population of finite size $k$, with any one-to-one infectious contact rate $\lambda$ and  removal rate $\gamma$. The results on within-household or within-workplace epidemics will follow by choosing these parameters accordingly. 

More precisely, consider the continuous-time Markov chain $X_{k, \lambda, \gamma} = \left(S, I\right)$ taking values in $\bar{\Omega}(k) = \{(s,i) \in \left(\N \cup \{0\}\right)^2: s + i \leq k\}$, and whose transition rates are given by 
\begin{equation}
\label{rate-sir}
\begin{aligned}
\text{Transition} &\quad & \text{Rate} \\
(s, i) \to (s - 1, i + 1) & \quad & \lambda s i; \\
(s, i) \to (s, i-1) & \quad & \gamma i.
\end{aligned}
\end{equation}
Then $S_t$ and $I_t$ represent respectively the number of susceptible and infected individuals at time $t$. Furthermore, the initial condition of interest is $\xx = (k-1, 1)$, and $\P_{\xx}$ denotes the probability conditionally on $X_{k, \lambda, \gamma}(0) = \xx$. 

Let us start by summarizing the results obtained by \cite{pellisEpidemicGrowthRate2011} on this matter. From (\ref{rate-sir}), it is obvious that when the population is in state $(s,i)$, a new infection takes place at rate $\lambda s i$. In particular, this rate is non-null if and only if $si \geq 1$, so we can restrict the study to the set of transient states $\Omega(k) = \{(s,i) \in \bar{\Omega}(k): i \geq 1\}$. Let $\zeta_{k, \lambda, \gamma}(t)$ be the average infection rate  in a population of composition $X_{k, \lambda, \gamma}(t)$, conditionally on $X_{k, \lambda, \gamma}(0) = \xx$. It is clear that, by definition,
\begin{equation}
\zeta_{k, \lambda, \gamma}(t) = \sum_{(s,i) \in \Omega(k)} \lambda s i \P_{\xx}(X_{k, \lambda, \gamma}(t) = (s,i)).
\end{equation}

\noindent Consider $\Qlg(k)$ the restriction of the generator of $X_{k, \lambda, \gamma}$ to $\Omega(k)$, which is defined as the following matrix indexed by states in $\Omega(k)$\footnote{We have made a slight change here compared to  \cite{pellisEpidemicGrowthRate2011} since the mortality matrix
$\Delta$ needed to be deleted from the expression of $Q$.}

\begin{equation}
\begin{aligned}
\label{eq-Q}
\forall (s,i)& \in \Omega(k), \forall (s',i') \in \Omega(k),  \\ 
&\left(\Qlg(k)\right)_{(s,i),(s',i')} =
\begin{cases}
\begin{aligned}
- \lambda s i - \gamma i \;\; & \text{ if } (s',i') = (s,i); \\
\lambda s i \;\; & \text{ if } (s',i') = (s-1, i+1); \\
\gamma i \;\; & \text{ if } (s',i') = (s, i-1); \\
0 \;\; & \text { otherwise.}
\end{aligned}
\end{cases}
\end{aligned}
\end{equation}

\noindent Then it is well known that for all $(s,i)$ in $\Omega(k)$,
\begin{equation*}
\P_{\xx}(X_{k, \lambda, \gamma}(t) = (s,i)) = \left(\rme^{t \Qlg(k)}\right)_{\xx,(s,i)}.
\end{equation*} 
Thus, a  computation readily yields the following Laplace transform of $\zeta_{k, \lambda, \gamma}$, where $I_{d(k)}$ is the identity matrix of appropriate dimension, namely $d(k) = \#\Omega(k) = k(k+1)/2$: for any $u \geq 0$,
\begin{equation}
\label{eq-lapzeta}
\Lap{k}{\lambda}{\gamma} (u) \coloneqq \Laplace \left( \zeta_{k, \lambda, \gamma} \right)(u) = \sum_{(s,i) \in \Omega(k)} \lambda si \left(\left(uI_{d(k)} - \Qlg(k)\right)^{-1}\right)_{\xx,(s,i)}.
\end{equation}

Let us now turn to our contributions. As we will see in the following proposition, we show that it is possible to go one step further and give an analytic expression for the relevant coefficients of
$$\hQ_{k, \lambda, \gamma} (u) \coloneqq \left(uI_{d(k)} - \Qlg(k)\right)^{-1},$$
for any population size $k$. This will finally allow us to give a more explicit expression of $\Lap{k}{\lambda}{\gamma}$.

Studying the restricted generator $\Qlg(k)$ leads to consider possible trajectories in $\Omega(n)$ leading from some state $(k-\ell, \ell)$ to another state $(s, i)$, for $\ell \leq k$ and $i \leq m = s+i \leq k$ (see Appendix \ref{appendix-proofs}). This incites us to introducing the following set, which allows to encode this set of trajectories:
\begin{equation}
\begin{aligned}
\mathcal{I}_k (\ell,m,i) = \big\{
(i_0,& \dots, i_{m+1}) \in \{\ell\}\times \N^{m} \times \{i\}: \\ \;\;
    & i_m \leq i,\; i_{j-1} - 1 \leq i_j  \leq k-j \quad \forall 1 \leq j \leq m\big\}.
\end{aligned}
\end{equation}
We are now ready to state our result.

\begin{prop} 
\label{prop-hatq}
Let $\ell \in \{1, \dots, k\}$ and consider $(s,i) \in \Omega(k)$ such that $s \leq k-\ell$. \\ 
Then for any $u \geq 0$,
  \begin{equation}
        \label{eq-hatq}
       	\left(\hQ_{k, \lambda, \gamma} (u)\right)_{(k-\ell,\ell),(s,i)} = \frac{1}{u + \lambda si + \gamma i} \sum_{ \mathfrak{i} \in \mathcal{I}_k(\ell,m,i)}     
         \prod_{j=0}^{m} q_{k, \lambda, \gamma}(\mathfrak{i},j; u) g_{k, m,\lambda, \gamma}(\mathfrak{i},j; u)
        \end{equation}
 where $m = k - (s + i)$ and 
     \begin{equation}
     \label{eq-qn}
     q_{k, \lambda, \gamma}(\mathfrak{i},j; u)=\prod_{w=i_j}^{ i_{j+1} - \ind_{\{\ell = m\}} }\left[1+\frac{u+\gamma w}{\lambda (k-j-w)w}\right]^{-1}
     \end{equation}
     and
      \begin{equation}
      \label{eq-gn}
      g_{k,m, \lambda, \gamma}(\mathfrak{i},j; u) = 
      \begin{cases} 
      \begin{aligned}
      & \left[1+\frac{u + \lambda(k - j - i_{j+1} - 1)(i_{j+1} +1)}{\gamma(i_{j+1}+1)}\right]^{-1} \;\; & \text{ for } j < m, \\
       & 1 \;\; & \text{ for } j = m.
  	\end{aligned}
       \end{cases}
      \end{equation}
     Furthermore, for every state $(s,i) \in \Omega(k)$ such that $s > k-\ell$, $\left(\hQ_{k, \lambda, \gamma} (u) \right)_{(k-\ell,\ell),(s,i)} = 0$.
\end{prop}

The proof of Proposition \ref{prop-hatq} uses arguments of linear algebra. Details can be found in Appendix \ref{appendix-proofs}. Using Equation (\ref{eq-lapzeta}), a more explicit expression of $\Lap{k}{\lambda}{\gamma}$ follows from equation (\ref{eq-hatq}). Let us define the ensemble
        \begin{equation*}
        \begin{aligned}
        \mathcal{I}_{k}(m) = \big\{&
        (i_0, i_1, \dots, i_{m},i_{m+1}) \in \{1\}\times \N^{m+1}: \\
        & i_{m+1} \leq k - m,\,  i_m \leq i_{m+1},\, i_{j-1} - 1 \leq i_j   \quad \forall 1 \leq j \leq m \big\}.
        \end{aligned}
  \end{equation*}
We can now state the result, using the same notations as in Proposition \ref{prop-hatq}.

\begin{corol}
\label{prop-lapzeta}
For any integer $k$ and any set of parameters $\lambda, \gamma > 0$ and any $u \geq 0$,
        \begin{equation}
        \label{eq-lapn}
        \Lap{k}{\lambda}{\gamma} (u) = \sum_{m=0}^{k-1} \sum_{ \mathfrak{i} \in \mathcal{I}_{k}(m)}    
         c_{k, \lambda, \gamma}(\mathfrak{i}; u) \prod_{j=0}^{m} q_{k, \lambda, \gamma}(\mathfrak{i},j; u) g_{k,m, \lambda, \gamma}(\mathfrak{i},j; u)
        \end{equation}
       
     	 with

	\begin{equation}
	c_{k, \lambda, \gamma} (\mathfrak{i}; u) = \left[1 + \frac{u +  \gamma i _{m+1}}{\lambda (k-m-i _{m+1})i _{m+1}}\right]^{-1}.
\end{equation}
\end{corol}
In the following section, we will see how $\Lap{k}{\lambda}{\gamma}$ intervenes in the computation of the growth rate $r$ of the epidemic. Another quantity of similar nature will be needed, namely the Laplace transform $\LapG{k}{\lambda}{\gamma}{\beta_G}$ of the average rate $\zeta_{k, \lambda, \gamma, \beta_G}^G$ at which all individuals of a structure of composition $X_{k, \lambda, \gamma}(t)$, conditionally on $X_{k, \lambda, \gamma}(0) = \xx$, contaminate individuals in the general population. Obviously, when the structure is in state $(s,i) \in \bar{\Omega}(k)$, this rate is given by $\beta_G i$, considering that the global proportion of susceptible individuals is close to one. In other words,

\begin{equation}
\label{eq-lapzetaG}
\LapG{k}{\lambda}{\gamma}{\beta_G} (u) \coloneqq \Laplace \left( \zeta_{k, \lambda, \gamma, \beta_G}^G \right)(u) = \sum_{(s,i) \in \Omega(k)} \beta_G i \left(
\hQ_{k, \lambda, \gamma} (u)
\right)_{\xx,(s,i)}.
\end{equation}

This rate is positive if and only if $(s,i) \in \Omega(k)$. Proceeding like before, we obtain the following formula:

\begin{corol}
\label{prop-lapzetag}
For any integer $k$, for any set of parameters $\lambda, \gamma, \beta_G > 0$, for any $u \geq 0$,
        \begin{equation}
        \label{eq-lapgn}
        \LapG{k}{\lambda}{\gamma}{\beta_G} (u)  = \sum_{m=0}^{k-1} \sum_{ \mathfrak{i} \in \mathcal{I}_{k}(m)}    
         c^\prime_{k, \lambda, \gamma, \beta_G} (\mathfrak{i}; u) \prod_{j=0}^{m} q_{k, \lambda, \gamma}(\mathfrak{i},j; u) g_{k,m, \lambda, \gamma}(\mathfrak{i},j; u)
        \end{equation}
        
with

\begin{equation}
c^\prime_{k, \lambda, \gamma, \beta} (\mathfrak{i}; u) = \frac{\beta i _{m+1}}{u + \lambda (k-m-i _{m+1})i _{m+1} + \gamma i _{m+1}}.
\end{equation}
\end{corol}

\begin{proof}[Proof of Corollaries \ref{prop-lapzeta} and \ref{prop-lapzetag}]
Start by noticing that 
\begin{equation}
\Omega(k) = \bigsqcup_{m=0}^{k-1} \bigsqcup_{i=1}^{k-m} \{(k-m-i,i)\}.
\end{equation}
Thus, equation (\ref{eq-lapzeta}) becomes, using Proposition \ref{prop-hatq}:
\begin{equation}
\Lap{k}{\lambda}{\gamma} (u) = \sum_{m = 0}^{k-1} \sum_{i=1}^{k-m} \frac{\lambda(k-m-i)i}{u + \lambda(k-m-i)i + \gamma i} \sum_{\mathfrak{i} \in \mathcal{I}_k(1,m,i)} \prod_{j=0}^m q_{k, \lambda,\gamma}(\mathfrak{i},j;u) g_{k,m, \lambda,\gamma}(\mathfrak{i},j;u)
\end{equation}
and the conclusion follows by the definition of $\mathcal{I}_k(m)$ and $c_{k, \lambda, \gamma}(\mathfrak{i};u)$.
Similarly, 
\begin{equation}
\LapG{k}{\lambda}{\gamma}{\beta_G} (u) = \sum_{m = 0}^{k-1} \sum_{i=1}^{k-m} \frac{\beta_G i}{u + \lambda(k-m-i)i + \gamma i} \sum_{\mathfrak{i} \in \mathcal{I}_k(1,m,i)} \prod_{j=0}^m q_{k, \lambda,\gamma}(\mathfrak{i},j;u) g_{k, m, \lambda,\gamma}(\mathfrak{i},j;u)
\end{equation}
and one concludes using the definition of $\mathcal{I}_k(m)$ and $c^\prime_{k, \lambda, \gamma, \beta_G}$.
\end{proof}

We now have introduced all necessary ingredients allowing for the computation of $r$, which is covered in detail in the next section.

\subsection{Characterization of the initial growth rate}
\label{sec-pellis}

Let us now turn to the characterization of $r$ for the multi-level model, which has been obtained by \cite{pellisEpidemicGrowthRate2011}. We summarize their arguments for the sake of completeness and in order to illustrate how Corollaries \ref{prop-lapzeta} and \ref{prop-lapzetag} complement their approach. 

\begin{figure}
	\centering
	\includegraphics[width=0.6\textwidth]{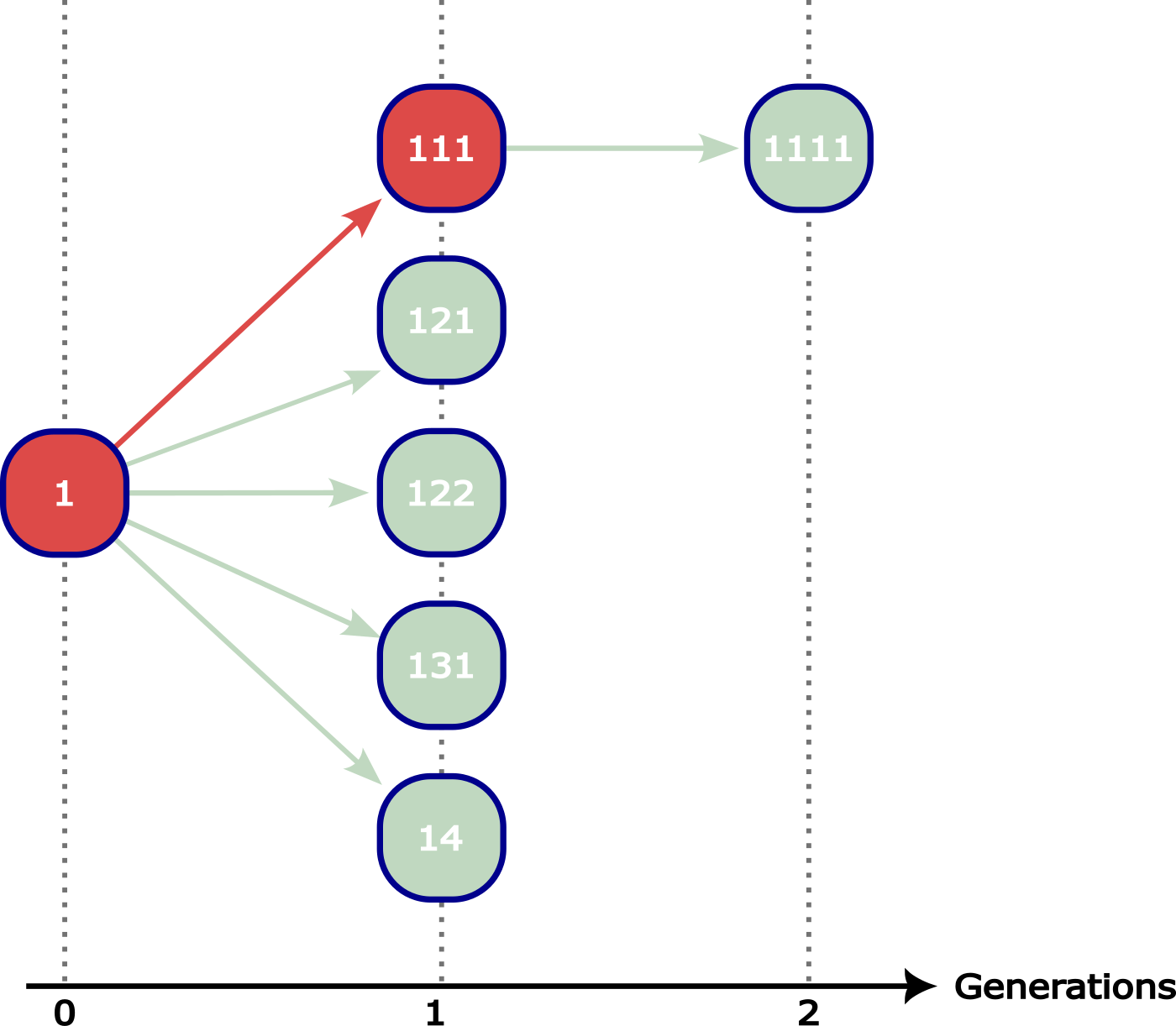}
	\caption{Illustration of the two-type household epidemic process. The example is the same as in Figure \ref{fig-branching-model}. The generations correspond here to generations of infected households, where the labels denote the first infected member of each household consistently with Figure \ref{fig-branching-model}. The colors of the arrows and of the households represent the type of infection, either globally (red) or locally (light green).}
	\label{fig-household-epidemic}
\end{figure}

Their main idea for computing the real time growth rate consists in considering the epidemic at the level of households instead of the individual level. Indeed,  it is possible to reduce the epidemic dynamics to a two-type process, distinguishing households that have been contaminated \emph{locally} (type $L$), if the first infected of the household contracted the disease at his workplace, or \emph{globally} (type $G$) otherwise. Remember that, during the early phase of an epidemic, every newly contaminated household or workplace will be fully susceptible except for its member who has just been infected. Thus, a household infects another household \emph{globally} whenever one of its contaminated members transmits the disease through the general population. Local transmission, on the other hand, occurs in the following way. Every time a member of a contaminated household $H_1$ is infected during its within-household epidemic, a within-workplace epidemic is started at his workplace. Then again, each coworker who is infected during the within-workplace epidemic introduces the disease into his household, which is regarded as \emph{locally} contaminated by household $H_1$. A slight subtlety is worth noticing: whether the first infected individual of a household participates in locally contaminating other households depends on the way he has been infected. If he has contracted the disease at his workplace, then he is not hold responsible for the the within-workplace epidemic there, and thus the other households that are infected through his workplace are not considered locally infected by his household. However, the opposite happens if he was infected through the general population, because he then launches a new within-workplace epidemic. This two-type process is depicted in Figure \ref{fig-household-epidemic}.

Suppose now that the epidemic is in its exponential growth phase, meaning that there exists a growth rate $r$ such that the number of infected individuals at time $t$ is proportional to $\rme^{r t}$. \cite{pellisEpidemicGrowthRate2011} argue that in this case, the household epidemic also grows exponentially at the same rate $r$. It thus is enough to study the previously introduced two-type process. 

For $x, y \in \{L,G\}$, let $\zeta_{xy}(t)$ denote the average rate at which a household of type $y$ which was infected $t$ units of time ago contaminates other households either locally if $x = L$, or globally if $x = G$. For $u \geq 0$, consider the matrix 
\begin{equation}
\label{Kmatrix}
K(u) = \begin{pmatrix} 
\Laplace \zeta_{GL} (u) & \Laplace \zeta_{GG} (u) \\
\Laplace \zeta_{LL} (u) & \Laplace \zeta_{LG} (u) 
\end{pmatrix}.
\end{equation}
It is a classical result \cite{diekmannMathematicalEpidemiologyInfectious2000, pellisEpidemicGrowthRate2011} that for this two-type setting, the growth rate $r$ is characterized as being the unique solution of the implicit equation 
\begin{equation}
\label{eq-r}
\rho(K(r)) = \frac{1}{2} \left(\Trace(K(r)) + \sqrt{\Trace(K(r))^2 - 4\det(K(r))}\right)= 1,
\end{equation}
where the operators $\rho$, $\Trace$ and $\det$ denote the spectral radius, trace and determinant, respectively. 

It thus only remains to take a closer look at $\zeta_{xy}(t)$ for all $x, y \in \{L,G\}$. As these are average rates of infection, one has to take into account the probability for a newly infected individual to belong to a household or workplace of a given size. Naturally, as households and workplaces are chosen independently uniformly at random, size-biased distributions appear both in the case of global and local infections. This leads us to introduce the following notation. For any application $f: (n,z) \mapsto f_n(z)$ on $\N \times \mathbb{R}$ and measure $\nu$ on $\N$, $\nu(f_\bullet)$ defines the function on $\mathbb{R}$ such that $\nu(f_\bullet): z \mapsto \sum_n \nu(n) f_n(z)$.

Within a household of size $k$, by definition, the average rate at which global transmissions occur is $\zeta_{k, \lambda_H, \gamma, \beta_G}^G$.  Since a newly locally or globally contaminated household is of size $k$ with probability $\hat{\pi}^{H}_k$, it follows that 
\begin{equation}
\zeta_{GG} = \zeta_{GL} = \hat{\pi}^H(\zeta_{\bullet, \lambda_H, \gamma, \beta_G}^G).
\end{equation}

On the other hand, within a globally contaminated household of size $n$, new cases appear at rate $\zeta_{k, \lambda_H, \gamma}$. Once infected, each of these individual transmits the disease to his coworkers on average at rate $\zeta_{w, \lambda_W, \gamma}$, where $w$ is the appropriate workplace size. This also applies to the globally infected individual who launched the within-household epidemic. If the household is contaminated locally, the reasoning is the same, except that the initially infected member is not hold responsible for his workplace epidemic, as he is himself a secondary case only. As a consequence, 
\begin{equation}
\begin{aligned}
\zeta_{LG} &=  \hat{\pi}^W(\zeta_{\bullet, \lambda_W, \gamma}) + \hat{\pi}^H(\zeta_{\bullet, \lambda_H, \gamma})* \hat{\pi}^W(\zeta_{\bullet, \lambda_W, \gamma}) \\
\zeta_{LL} &=  \hat{\pi}^H(\zeta_{\bullet, \lambda_H, \gamma})* \hat{\pi}^W(\zeta_{\bullet, \lambda_W, \gamma}).
\end{aligned}
\end{equation}

Using standard properties of Laplace transforms, $K(u)$ thus admits the following expression as derived by \cite{pellisEpidemicGrowthRate2011}, where the coefficients can now be computed using Corollaries \ref{prop-lapzeta} and \ref{prop-lapzetag}: 
\begin{equation}
\label{eq-Kr}
K(u) = \begin{pmatrix} 
 \hat{\pi}^H \left(\LapG{\bullet}{\lambda_H}{\gamma}{\beta_G}\right) (u) & \hat{\pi}^H \left(\LapG{\bullet}{\lambda_H}{\gamma}{\beta_G}\right) (u) \\
\hat{\pi}^H \left(\Lap{\bullet}{\lambda_H}{\gamma}\right) (u)  \hat{\pi}^W \left(\Lap{\bullet}{\lambda_W}{\gamma}\right) (u) &  \left(1 + \hat{\pi}^H \left(\Lap{\bullet}{\lambda_H}{\gamma}\right) (u)\right)  \hat{\pi}^W \left(\Lap{\bullet}{\lambda_W}{\gamma}\right) (u)
\end{pmatrix}.
\end{equation}
Numerical methods then allow to solve the implicit equation (\ref{eq-r}) for the exponential growth rate $r$. We refer to Appendix \ref{supp:nuemrical_methods} for details on the computation procedure.

\subsection{ODE reduction of the multilevel model based on the initial growth rate}
\label{sec-reduction}

The reduced model is a standard deterministic \emph{SIR} model, with infection rate derived from the growth rate $r$, obtained from Equation \eqref{eq-r}.
The reduction of the model is defined by the following set of ordinary differential equations:
\begin{equation}
\label{eq:reduction}
\left\{
\begin{aligned}
\frac{dS}{dt} = & - (r + \gamma) S I \\
\frac{dI}{dt} = & (r + \gamma ) S I - \gamma I \\
\frac{dR}{dt} = & \gamma I \\
\end{aligned}
\right.
\end{equation}

The prediction accuracy of the structured model by the reduced model is illustrated in Figure \ref{fig:reduction_example}, where 1\% of individuals are initially infected. In this example, we compare simulated epidemics of the stochastic structured model with scenario 1 from Table \ref{tab:scenfeat} and reference household and workplace size distributions to the ODE reduction from Table \ref{tab:dist-bis}. The reduced model, based on the initial growth of infected population, accurately predicts, as expected, the early stages of the epidemic. We note however, that the prediction accuracy decreases with time and that the epidemic peak and the final size are overestimated by the reduced model.

Further explorations, using all scenarios and exploratory workplace distributions A, indicate the same trends.
Comparison of epidemic outcomes such as the epidemic peak and the final size, between simulations of the stochastic structured model and their reduced counterpart are presented in Figure \ref{fig:reduction_robustness}. We observe that, regardless of the scenario, the peak size and final epidemic size are largely correlated between the reduced model and numerical simulations of the stochastic structured model. A slight tendency to overestimate the epidemic outcomes with the reduced model can also be noticed. However, Figure \ref{fig:reduction_robustness} (top) illustrate that the epidemic outcomes of the reduced model remain close to the simulations of the stochastic structured model, with differences from the exact model simulations of less than 5\% of the total number of individuals with the reduced model.  The overestimation is more visible concerning the final epidemic size. This figure also shows that the epidemic parameters influence the prediction quality with the reduced model.
This figure also illustrates the effect of the values of $r$ and $p_G$ on the prediction of the final epidemic size with the reduced model. Higher values of $r$, combined with low value for $p_G$ (or high $p_W$), such as in scenario 1 and 2, tend to decrease the prediction accuracy with the reduced model. For lower values of the growth rate, which also corresponds to lower epidemic final size, the prediction accuracy is high, see for example scenarios 6 to 11.

The quality of the prediction is expected to be good for high values of $p_G$, as for $p_G=1$ the reduced model is strictly equivalent to the deterministic approximation of the stochastic \emph{SIR} model without structures. As the value of $p_G$ decreases, propagation of the epidemic through households and workplaces becomes dominant and the structured model cannot be approximated by a uniformly mixing deterministic \emph{SIR} model regardless of the growth rate, as illustrated in Figure \ref{fig:reduction_robustness} (bottom). This figure shows that the epidemic outcomes for the stochastic structured model cannot be obtained by any deterministic \emph{SIR} model for scenarios with lower values of $p_G$ such as scenario 1 and 2, which deviate from the black line representing the possible outcomes with a deterministic \emph{SIR} model. In addition, Figure \ref{fig:reduction_robustness} (top) shows that, for scenarios with lower epidemic size, the prediction of epidemic outcomes by the reduced model is more accurate. This figure also illustrates that the peak size is accurately predicted by the reduced model in these cases, while the prediction of the final size is less accurate. This provides further evidence that the reduction captures the early phase of the epidemic.

\begin{figure}[!ht]
\centering
\includegraphics[width=0.7\textwidth]{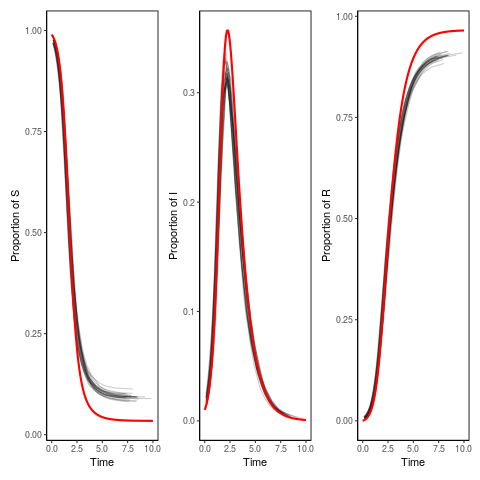} 
\caption{Evolution of the proportion of S, I, R individuals in 30 runs of the stochastic structured model (black), and the reduced model (red). Simulations are performed in a population with 100,000 individuals, with the reference household and workplace size distributions and epidemic parameters from scenario 1. The origin of time of the numerical simulation has been set to the time where 1\% of individuals are infected. The simulation of the reduced model is performed with 1\% of initially infected individuals.}
\label{fig:reduction_example}
\end{figure}

\begin{figure}[!ht]
  \centering
\includegraphics[width=0.7\textwidth]{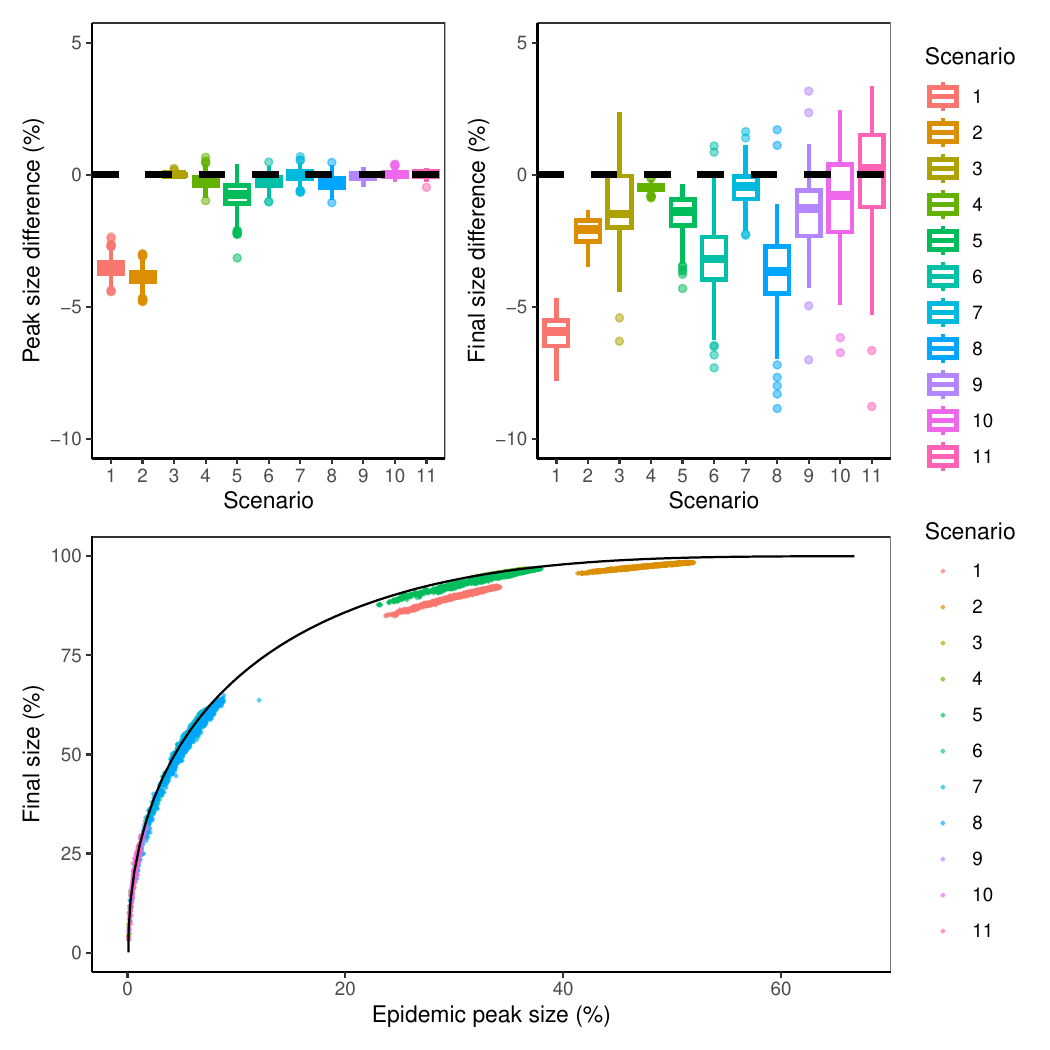}
\caption{Boxplot of the differences in the peak sizes (top left) and final sizes (top right) between the complete and the reduced models, obtained from simulations of the stochastic structured model and the corresponding reduced model. Simulated final size as a function of the peak size is reported (bottom), the black line represents the value for the \emph{SIR} model \eqref{eq:reduction}. Simulations of the stochastic structured model are performed with population size of 100,000, reference household distribution, workplace size distributions set A, for all scenarios from Table \ref{tab:scenfeat} (point color). Each combination of scenario and workplace size distribution from set A is repeated 10 times. Only simulations where an epidemic outbreak occurred (\emph{i.e.} more than 3\% of the population become infected) are reported in this figure.}
\label{fig:reduction_robustness}
\end{figure}

\subsection{Assessment of the reduction robustness}
\label{robust:reduction}

So far, in the modeling  and for the simulations and computations, we have restricted ourselves to Markovian models and to the \emph{SIR} structure. The Markovian assumption simplifies the parametrization: there is one single rate of infection for each of three infection ways and one single rate of recovery. This assumption is not necessary and most of our results can be formalized in more general contexts. Similarly, the \emph{SIR} model is the canonical one, but models with different structures can be envisaged. We have checked by simulation that the reduction procedure based on the initial growth rate, as defined in Section \ref{sec-reduction} and in Appendix \ref{apx:seir-num}, still provides good predictions in more general contexts. In Figures S8 and S9 of the Online Resource, we illustrate through two examples the prediction of epidemic outcomes between the structured model with Gamma distributed individual recovery times and the corresponding reduced deterministic \emph{SIR} model. We can make similar observations as for the Markovian model that the reduced model generally provides a good prediction of the epidemic outcomes, even though the predictions are less accurate in some cases. Notice that the loss of accuracy was to be expected, as the reduced model differs both in terms of contact structure, and by the choice of the distribution of the infectious period (exponential distribution instead of Gamma distribution). We also show with an example provided in Figures S10 and S11 of the Online Resource that the reduction procedure provides good predictions of epidemic outcomes for an \emph{SEIR} model, where an additional latent (and non infectious) $E$ state is added to the \emph{SIR} model. Details on the growth rate derivation for the \emph{SEIR} model are provided in Appendix \ref{sec-SEIR-growthrate}, while computational aspects are addressed in Appendix \ref{apx:seir-num}.

\section{Discussion}
\label{discussion}

In this work, we study the effect of the size of closed structures (households, workplaces, schools...) on the propagation of epidemics, when individuals belong to different structures, chosen independently for a given individual. We assume that there is no dependence between the size of households and the size of workplaces. Our motivation comes in particular from the fact that control policies allow to change the distribution of structure sizes in various ways, for example by reducing the size of workplaces by teleworking or the size of schools by their (partial) closure. Optimizing control measures in terms of sanitary outcome has become a major challenge. Measuring the link between a social organisation, in terms of distribution of structure sizes, and the epidemic outcomes is a delicate issue. Indeed, explicit formula and low dimensional large population approximations are known only in uniformly mixing populations. These results concern explicit values for $R_0$ and initial growth rate $r$, simple equation for the total size of the outbreak, reduction to an \emph{SIR} three dimensional ODE with two parameters. Formulae become more complex when adding a local level consisting in one layer of structures (households), and even intractable with an additional structure like workplaces, where the chain of transmission does not need mean field infection any longer to propagate.

We have used both existing approaches (namely \cite{pellisEpidemicGrowthRate2011}) and added new developments, by providing more explicit expressions (Section \ref{reductionODE}) and simulations to study the role of structure sizes on the key outcomes of epidemics. Notice that even though the results of Corollaries \ref{prop-lapzeta} and \ref{prop-lapzetag} have not been used in this paper for evaluating the growth rate of simulation scenarios (details on computations in Appendix \ref{annex_numer}), their numerical implementation would have been an alternative method for achieving these computations. In our setting, this would not be numerically pertinent as workplace sizes can be relatively large and the formulae obtained in Corollaries \ref{prop-lapzeta} and \ref{prop-lapzetag} need to iterate over all elements in a set whose cardinal grows exponentially with the structure size. However, making use of these results may be pertinent when computing the growth rate in models considering households only, as those typically are of smaller sizes. 

We have focused on the Markovian case, \emph{i.e.} time of infection exponentially distributed and constant infection rate, and a simple \emph{SIR} structure, with local infection rates proportional to the number of susceptible and infected individuals. This basic framework was complemented by a robustness study on more complex settings (Sections \ref{robustinteract} and \ref{robust:reduction}). According to our findings, the structure size distribution plays a role on the key outcomes in most scenarios. More precisely, for a given number of structures and a given number of individuals and thus for fixed average structure size, the way individuals are distributed has a quantitative impact on the growth rate of infections, the total number of infected individuals and the size of the infected peak. In this setting, the variance of the structures size distribution provides a good proxy of this impact.

This finding may be related to previously known results on the importance of the variance of the degree distribution in configuration model settings. Indeed, \cite{brittonGraphsSpecifiedDegree2007} have pointed out the impact of the degree distribution variance on the reproduction number for \emph{SIR} epidemics on configuration graphs. Similarly, \cite{maEffectiveDegreeHousehold2013} have studied the case of a model with two levels of mixing, corresponding to a layer of households and a general population taking the form of a configuration graph. They have shown that in this case, the variance of the degree distribution in the general population has a strong influence on epidemic dynamics. In the case of the household-workplace model, we can consider the epidemic at the household level. In this situation, the size distribution of the workplace plays a crucial role in determining the number of households to which a given household is directly connected. Although the framework is clearly more complex than a simple configuration model, the distribution of workplace size may play a role similar in spirit to that of the degree distribution in the previously mentioned models based on configuration graphs.

As for the limitations of our study, the robustness analyses that we have carried out (Sections \ref{robustinteract} and \ref{robust:reduction}), even if they are rather summary, point to robustness of the main results when certain assumptions are modified. A more detailed analysis using sensitivity analysis is left for a future work.\\
On the one hand, the effect of the structure size distribution may be overestimated by the fact that we consider a linear infection rate. Indeed, the rate at which a susceptible individual is infected at a given time is assumed to be proportional to the number of infected individuals at the same time. This probably overestimates the real infection rate. We consider here structures which form a partition of the population, with uniform mixing within each structure. We think this assumption is rather relevant for households but it could be improved and extended for workplaces (and schools), where different levels of mixing could be considered (services/departments within companies, classes within schools, etc). The simulation study with sublinear infection rates in households and workplaces shows that the observation on the linear impact of variance on the final epidemic size is still valid (Figure \ref{fig:Variance_sqrt1} and Figure S5 of the Online Resource).  With sublinear infection rates, we also observed similar results regarding the robustness of the prediction of epidemic outcomes using the reduced model, namely that the reduction is accurate for high values of $p_G$ and a smaller epidemic size, see Figure S12 of the Online Resource. We note however, that the accuracy of the prediction using the reduced model is lower with sublinear infection rates, and the early epidemic is not predicted as accurately.
When it comes to policies and controls such as teleworking, a notion of effective size and/or effective infection rate should probably be introduced, which is one of the interesting perspectives. \\
On the other hand, the effect of the structure size distribution is observed when looking at the characteristics of the epidemic. Variance arises using the size biased law and the fact that mean final size of epidemics is comparable to the size structure. Furthermore, following preliminary explorations, we observed that the dependence in the structure size distribution may imply a higher moment of it than the second moment linked to its variance. This could confer a greater impact of large structures than the variance would predict. This reinforces the importance of the structure size distribution and specifically the importance of the health benefit of moving toward structures (workplace offices, school classrooms...) of the same size. Variance remains the simplest proxy we have found in general, but further exploration may reveal another form of dependence.

Thus, based on our results, the variance of the structure size distribution, and more generally the ratio of the second to the first moment when the latter or the number of structures is not fixed, provides a good indicator to measure the impact of the distribution of individuals within small social structures. Nevertheless, much work remains to be done to identify the key parameters related to population structure that determine epidemic outcomes.

In this study, we are also interested in model reduction, in order to have a parsimonious model that is sufficiently accurate in terms of prediction and fast to run. This is an important issue, especially in a context where many scenarios and control policies need to be evaluated. Thus, our goal was to reduce the model to only a few parameters and variables. 
We obtained that by using the initial infection growth rate, which keeps track of the contact structure in a subtle (almost explicit but rather complex) way. We then make use of a classical \emph{SIR} model (three variables and two parameters) as a reduction of the initial stochastic individual-centered model with three types of contacts.

In summary, our study highlights that knowledge and modeling of the size of contact structures appears to be important in characterizing the outcome of an epidemic. It also points to the major role of the initial growth rate of the infection, which is unique, contrary to the reproduction number which can have several interpretations in multilevel contact models, and also difficult to treat.  We have tried to provide a more explicit expression for the initial growth rate, supplementing the literature. This allowed us first to conduct the study of the impact of the structure size distribution. As we already know, the initial growth rate provides the extinction/outbreak criterion by its sign and the rate of progression. It is related to the peak and timing of the peak, when the epidemic dynamics explode, and to the extinction rate when this dynamics decline. The initial growth rate has also been a key input to the reduction problem, allowing a complex contact structure to be reduced to a few parameters, while retaining the essence of the qualitative and quantitative behavior of the epidemic beyond the initial time.

\begin{appendices}

\section{Generation of structure size distributions}
\label{appdx:sizedist}

Let us start by giving detail on the way the reference workplace size distribution is derived from the INSEE French workplace size distribution of 2018. Indeed, the reference workplace size distribution is obtained from the number of workers in workplace size classes provided by INSEE. On the one hand, workers belonging to workplaces in size classes lower than 50 are placed, for our workplace size distribution, in workplaces with size uniformly chosen in the size range covered by the size class. On the other hand, workers belonging to workplaces of over 50 employees are arbitrarily placed in workplaces of size 50 in our workplace distribution, as we assume uniform mixing within workplaces which becomes unrealistic for very large workplace sizes.

Finally, in order to generate distributions with given mean $m$, variance and maximum size $n_{\max}$, we proceed as follows. First, we create a set $\mathcal{P}$ of distributions which each charge only two sizes in $\{1, \dots, n_{\max}\}$ and mean $m$. In other words, for any $1 \leq k < m < k' \leq n_{\max}$, we define $p = c\delta_k + (1-c)\delta_{k'}$ where $c = (k'-m)/(k'-k)$. 
Second, we construct a new set of distributions $\mathcal{D}$ from mixtures of two distributions of the previous set, such that the resulting distribution has a given variance. An element $d \in \mathcal{D}$ hence is obtained by taking $p_1, p_2 \in \mathcal{P}$ and letting 
$d = c p_1 + (1-c) p_2,$
where the weight $c$ is chosen such that $d$ has the prescribed variance.
A target distribution $D$ is then generated by a random mixture of those elementary distributions: 
    \[D = \sum_{d \in \mathcal{D}} w_d d,\]
where $w_d$ are random weights such that $\sum_{d \in D} w_d = 1$. 

\section{Parameter values for scenarios}
\label{annex_numer}

The size distribution of households and workplaces combined with structure-dependent infection rates determines a value for the initial growth rate of the epidemic $r$, as well as probabilities of infections corresponding to the three sources, global mixing $p_G$, households $p_H$ and workplaces $p_W$.
Due to the constraints imposed by the structure size distributions in the structured epidemic model, it is not always possible to find numerical values of infection rates that lead to given values of growth rate $r$ and infection probabilities $p_G$, $p_H$ and $p_W$. 
Parameter selection for scenarios in Table \ref{tab:scenfeat} was performed, for the reference household and workplace distributions, using an optimisation procedure that yields infection rates leading to a solution for growth rate and infection probabilities values as close as possible to the target values. It relies on a cost function based on the mean square error between the target values and the trial values of $r$, $p_G$, $p_H$ and $p_W$. A hyper-parameter controls the importance given to the error on the growth rate.

Table \ref{tab:scenarios} summarizes the values of $r$, $R_I$, $p_G$, $p_H$ and $p_W$ for each scenario of Table \ref{tab:scenfeat}, as well as the obtained values of epidemic parameters for the reference structure size distributions.

Scenarios provided in Table \ref{tab:scenfeat} combined with the workplace distributions from Table \ref{tab:dist-bis} allow the exploration of the relevant behavior of the structured epidemic model, covering a wide range of epidemic settings. 
The epidemic final size for each scenario and workplace distribution is reported in Figure S1 of the Online Resource.

As mentioned in Section \ref{sec-param}, our scenarios correspond to, or are close to, realistic epidemic settings for three diseases of interest: influenza (\cite{ajelliRoleDifferentSocial2014}), COVID-19 (\cite{locatelliEstimatingBasicReproduction2021, galmicheExposuresAssociatedSARSCoV22021}) and chickenpox (\cite{silholModellingEffectsPopulation2011}). Notice that while comparing proportions of infection at the local and global level is straightforward, the task is more delicate for reproduction numbers. 
As \cite{ajelliRoleDifferentSocial2014} and \cite{locatelliEstimatingBasicReproduction2021} infer the reproduction number based on the exponential growth rate, we have followed this approach and based the comparison on the reproduction number $R_0$ defined as $R_0 = 1 + r/\gamma$, where $r$ is the exponential growth rate (\cite{trapmanInferringR0Emerging2016}).Of course, \cite{ajelliRoleDifferentSocial2014}, \cite{locatelliEstimatingBasicReproduction2021} and \cite{silholModellingEffectsPopulation2011} do not use this precise definition of $R_0$, but it seemed the best compromise for comparison as it is closest in spirit to the the studies on influenza and COVID-19, while the study on chickenpox unfortunately does not detail their definition of $R_0$. The results are shown in Figure \ref{fig:R0pG-scenarios}, which illustrates that our procedure covers a wide range of realistic epidemic settings.

\begin{table}[t]
\begin{center}
\begin{minipage}{\textwidth}
\caption{Values of growth rate, reproduction number and proportions of infection per layer for each simulation scenario, as well as the corresponding epidemic parameters, in the case of the reference structure size distributions (recall that $\gamma = 1$ in all scenarios).}
\label{tab:scenarios}
\begin{tabular}{rrrrrr|rrr}
  \hline
  Scenario & Growth rate & $R_I$ & $p_H$ & $p_W$ & $p_G$ & $\lambda_H$ & $\lambda_W$ & $\beta_G$ \\ 
  \hline
      1 & 2.4822 & 2.5028 & 0.4217 & 0.1788 & 0.3995 & 11.852 & 0.009 & 1.000 \\ 
     2 & 4.9937 & 4.6876 & 0.2796 & 0.3471 & 0.3732 & 11.443 & 0.020 & 1.750 \\ 
     3 & 0.0009 & 0.9923 & 0.4070 & 0.3927 & 0.2002 & 0.376 & 0.010 & 0.199 \\ 
     4 & 2.5203 & 3.6460 & 0.1521 & 0.1472 & 0.7008 & 0.356 & 0.010 & 2.555 \\ 
     5 & 2.5017 & 5.2021 & 0.1301 & 0.5445 & 0.3254 & 0.451 & 0.028 & 1.693\\ 
     6 & 0.5054 & 1.5740 & 0.3868 & 0.3456 & 0.2676 & 0.653 & 0.012 & 0.421 \\ 
     7 & 0.5061 & 1.5187 & 0.1120 & 0.1122 & 0.7758 & 0.094 & 0.004 & 1.178  \\ 
     8 & 0.5020 & 1.5920 & 0.3975 & 0.4057 & 0.1968 & 0.730 & 0.013 & 0.313 \\ 
     9 & 0.1104 & 1.1330 & 0.3879 & 0.4182 & 0.1940 & 0.404 & 0.012 & 0.220 \\ 
    10 & 0.0900 & 1.0989 & 0.2631 & 0.2646 & 0.4723 & 0.196 & 0.007 & 0.519\\ 
    11 & 0.0706 & 1.1068 & 0.3449 & 0.5883 & 0.0668 & 0.309 & 0.016 & 0.074 \\ 
  \hline
\end{tabular}
\end{minipage}
\end{center}
\end{table}

\begin{figure}[!ht]
  \centering
\includegraphics[width=0.7\textwidth]{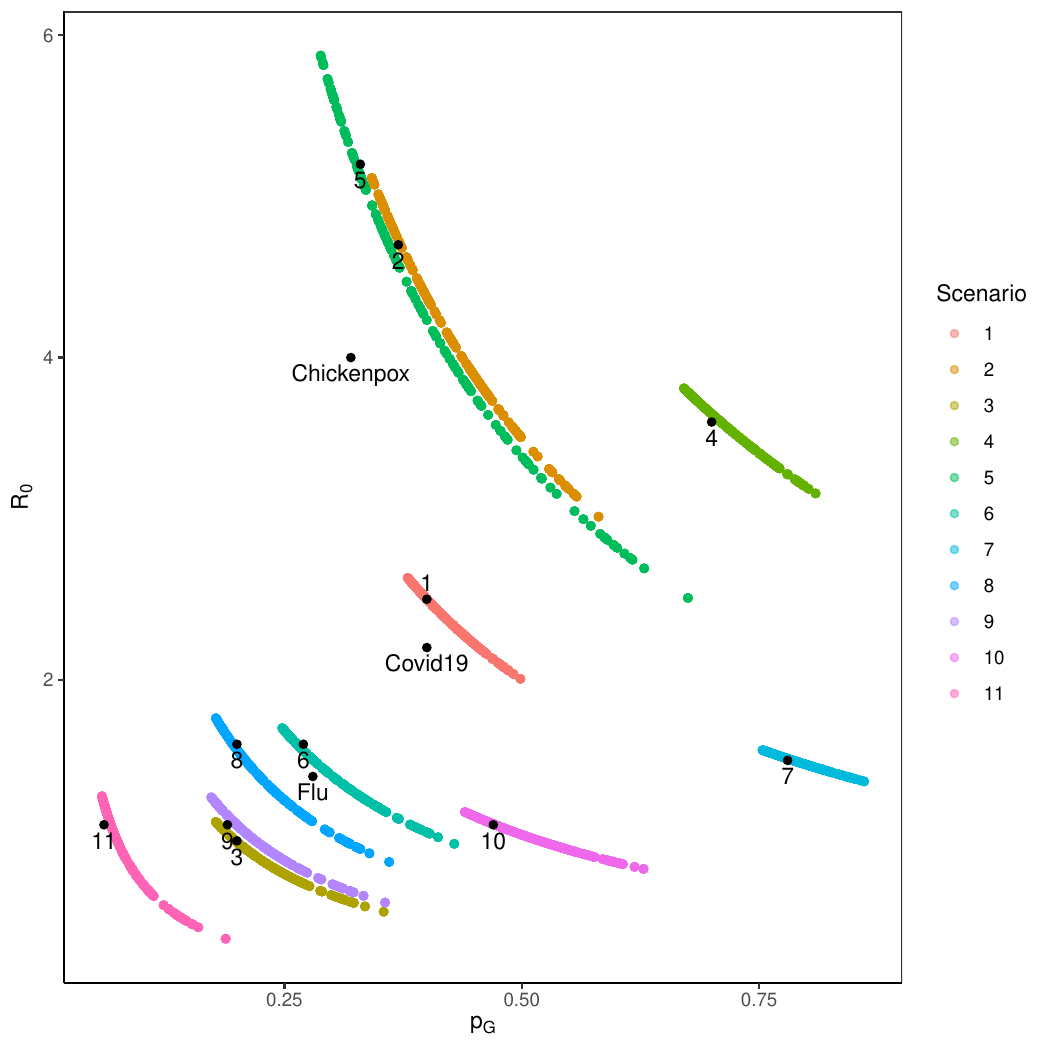}
\caption{Values of the reproduction number $R_0$, and proportion of infection via global mixing $p_G$ for all scenarios of Table \ref{tab:scenfeat} and all exploratory workplace size distributions from set A, Table \ref{tab:dist-bis}. The household size distribution is the reference household size distribution. For each scenario 1-11, labeled black points correspond to the values achieved for the reference workplace distribution. For comparison, values reported in the literature for COVID-19, influenza and chickenpox are also shown in black.}
\label{fig:R0pG-scenarios}
\end{figure}

\section{Numerical computations of epidemic parameters and outcomes}
\label{supp:nuemrical_methods}

The values of $R_I$, $p_G$, $p_H$ and $p_W$  are obtained from Equations \eqref{eq:reprod_matrix} and \eqref{eq-pgphpw}. They both require the values of $\mathcal I_G$, $\mathcal I_G$, $\mathcal I_W$ which are given by Equation \eqref{eq:nhat}. Evaluations of $\mathcal I_G$, $\mathcal I_G$, $\mathcal I_W$ require the values of $i_H(k)$ and $i_W(k)$ which we obtained from numerical simulations of the within structures epidemic. The growth rate $r$ can be obtained in several ways, which all involve some form of solution for Equation \eqref{eq-r}, which we solved using a root finding algorithm. Elements of Matrix \eqref{Kmatrix} can be obtained by numerical simulation of within structure epidemics, numerical integration and numerical matrix inversion. We also provide analytical formulations for the \emph{SIR} and \emph{SEIR} models, with linear infection rates, in Equations \eqref{eq-lapzeta} and forthcoming Equation \eqref{eq:lapzeta_seir}. Alternatively, we provide a fully analytical way to obtain the growth rate of the \emph{SIR} model with linear infection rates, in Corollaries \ref{prop-lapzeta} and \ref{prop-lapzetag}. 
The values of the epidemic peak and size are obtained by stochastic simulation of the structured model, and the values of the epidemic peak and size of the reduced model are obtained by simulation of Equation \eqref{eq:reduction}. For a summary of the numerical evaluation of epidemic quantities, see Table \ref{tab:numerical_qty}.

\begin{table}[ht]
\begin{center}
\begin{minipage}{\textwidth}
\caption{Numerical methods used to compute epidemic parameters and outcomes.}
\label{tab:numerical_qty}
\begin{tabular}{ll}
\toprule
Epidemic parameter/outcome & Numerical method  \\
\midrule
$\mathcal I_G$, $\mathcal I_G$, $\mathcal I_W$          &  Equation \eqref{eq:nhat}, with $i_H(k)$ and $i_W(k)$ obtained from \\
                                        & numerical simulations of the within structure epidemic.\\
$R_I$                                   & Largest eigenvalue of Equation \eqref{eq:reprod_matrix}.\\
$p_G$, $p_H$ and $p_W$                  & Equation \eqref{eq-pgphpw}.\\
Growth rate $r$                         & Equation \eqref{eq-r} with Laplace transforms for elements of matrix\\
                                        & \eqref{Kmatrix} which can be numerically evaluated. Alternatively, in \\ 
                                        &  the case of the \emph{SIR} (or \emph{SEIR}) model with linear infection rates, \\
                                        & these elements are defined in Equation \eqref{eq-lapzeta} (resp. \eqref{eq:lapzeta_seir}), where\\
                                        &  matrices in the sum are inverted numerically.\\
Peak and final size                     & Numerical simulation of the structured epidemic model or \\
      & numerical simulation of Equation \eqref{eq:reduction} for the reduced\\
      & epidemic model.\\
\bottomrule
\end{tabular}
\end{minipage}
\end{center}
\end{table}

\section{Proof of Proposition \ref{prop-hatq}}
\label{appendix-proofs}

The notations are the same as previously introduced in Section \ref{sec-laplace}. For two integers $p,q$, we further let $I_p$ denote the identity matrix of dimension $p$, and $M_{p,q}(\mathbb{R})$ (resp. $M_{p}(\mathbb{R})$) the space of $p\times q$ (resp. $p \times p$) matrices with real coefficients. The aim is to compute
\begin{equation}
\hQ_{k, \lambda, \gamma}(u) = (uI_{d(k)} - \mathbf{Q}_{\lambda,\gamma}(k))^{-1}
\end{equation}
for $u\geq 0$,
which allows us to obtain Corollaries \ref{prop-lapzeta} and \ref{prop-lapzetag}.

\begin{proof}[Proof of Proposition \ref{prop-hatq}.] 

We start by noticing that when $k = 1$, necessarily, $\ell=1$, $(s,i) = (0,1)$ and $m = 0$. The set of interest becomes $\mathcal{I}_1(1,0,1) = \{(1,1)\}$, and the right side of Equation (\ref{eq-hatq}) equals $(u + \gamma)^{-1}$. On the other hand, it is obvious that $\left(uI_1 - \mathbf{Q}_{\lambda,\gamma}(1)\right) = (u + \gamma)$. Thus, (\ref{eq-hatq}) holds when $k=1$.

Suppose now that Equation (\ref{eq-hatq}) is true for some integer $k$. We proceed by induction. Notice that it is possible to enumerate the states of $\Omega(k+1)$ in such a way that $\mathbf{Q}_{\lambda, \gamma}(k+1)$ is upper triangular. Indeed, it is enough to enumerate first all states in $\widehat{\Omega}(k+1) = \{(s,i) \in \Omega(k+1): s + i = k+1\}$ as $\{(k + 1 - \ell, \ell): 1 \leq \ell \leq k+1\}$, so that progression from one state to another occurs by infections which are not reversible as individuals are immune after infection. This process is repeated for states $\widehat{\Omega}(m) = \{(m-\ell,\ell): 1 \leq \ell \leq m\}$ for $m = k, \dots, 1$, so that transition from one set of states to the next occurs by removal of an infected, which also is irreversible as individuals will remain immune afterwards. 

This way of enumerating $\Omega(k+1)$ is particularly interesting as $\Omega(k+1) = \bigsqcup_{m = 1}^{k+1} \widehat{\Omega}(m)$, so that $\mathbf{Q}_{\lambda,\gamma}(k+1)$ can be regarded as the following block matrix:
\begin{equation*}
\mathbf{Q}_{\lambda,\gamma}(k+1) = 
\begin{pmatrix}
\bfA & \bfB \\
\mathbf{0} & \mathbf{Q}_{\lambda,\gamma}(k)
\end{pmatrix}.
\end{equation*}

Here, blocks $\bfA \in M_{k+1}(\mathbb{R})$ and $\bfB\in M_{k+1, d(k)}(\mathbb{R})$ represent events of infections and removals in $\widehat{\Omega}(k+1)$, respectively, where we recall that $d(k) = \#\Omega(k) = k(k+1)/2$. As previously, the elements of $\mathbf{Q}_{\lambda,\gamma}(k+1)$ can also be indexed by states in $\Omega(k+1)$. Naturally, $uI_{d(k+1)} - \mathbf{Q}_{\lambda,\gamma}(k+1)$ also is a block matrix, simply replacing the blocks $\bfA$ and $\mathbf{Q}_{\lambda,\gamma}(k)$ by $\bfA^\prime = uI_{k+1} - \bfA$ and $uI_{d(k)} - \mathbf{Q}_{\lambda,\gamma}(k)$ respectively. 

More precisely, $\bfA^\prime$ is an upper bidiagonal matrix such that
\begin{equation*}
\begin{aligned}
\forall \ell \in \{1, \dots, k+1\}, \quad
\bfA^\prime_{\ell,\ell} &= \bfA^\prime_{(k+1-\ell,\ell),(k+1-\ell,\ell)} = u + \lambda (k+1-\ell)\ell + \gamma \ell,  \\
\forall \ell \in \{1, \dots, k\}, \quad
\bfA^\prime_{\ell, \ell+1} &= \bfA^\prime_{(k+1-\ell,\ell),(k+1-(\ell+1),\ell+1)} = -\lambda (k+1-\ell)\ell.
\end{aligned}
\end{equation*}

\noindent Since $\bfA$ is upper diagonal, its inverse matrix is easily computable \cite{chatterjeeNegativeIntegralPowers1974}, and letting $a^{(k+1)}_{\ell,i} = (\bfA^\prime)^{-1}_{(k+1-\ell),(k+1-i,i)}$, we have: 
\begin{equation*}
a^{(k+1)}_{\ell,i} = \ind_{i \geq \ell} \frac{1}{u + \lambda(k+1-i)i + \gamma i} \prod_{j=\ell}^{i-1} \left(1 + \frac{u + \gamma j}{\lambda(k+1-j)j}\right)^{-1}.
\end{equation*}

Furthermore, as the only states of $\Omega(k)$ that are directly accessible by removal from $\widehat{\Omega}(k+1)$ belong to $\widehat{\Omega}(k)$, all coefficients of $\bfB$ are null except for the following:
\begin{equation*}
\bfB_{\ell,\ell-1} = \bfB_{(k+1-\ell,\ell),(k+1-\ell,\ell-1)} = -\gamma \ell, \quad \forall \ell \in \{2,\dots, k+1\}.
\end{equation*} 

Using the fact that $uI_{d(k+1)} - \mathbf{Q}_{\lambda, \gamma}(k+1)$ is a block matrix and that $\bfA^{\prime}$ and $uI_{d(k)} - \mathbf{Q}_{\lambda, \gamma}(k)$ are upper triangular with positive diagonal coefficients and thus invertible, we have the following:
\begin{equation*}
\hQ_{k+1, \lambda, \gamma} (u) = \left(uI_{d(k+1)} - \mathbf{Q}_{\lambda, \gamma}(k+1)\right)^{-1} =
\begin{pmatrix}
(\bfA^\prime)^{-1} & - (\bfA^\prime)^{-1} \bfB \hQ_{k, \lambda, \gamma}(u) \\
0 &  \hQ_{k, \lambda, \gamma}(u)
\end{pmatrix}.
\end{equation*}

\noindent Thus, we obtain that, for all $\ell, i \in \{1, \dots, k+1\}$ and $(s,i) \in \Omega(n)$,
\begin{equation}
\label{eq-inverse-q}
\begin{aligned}
\left(\hQ_{k+1, \lambda, \gamma} (u) \right)_{(k+1-\ell, \ell),(k+1-i,i)} & = a^{(k+1)}_{\ell,i} ; \\
\left(\hQ_{k+1, \lambda, \gamma} (u) \right)_{(k+1-\ell,\ell),(s, i)} & =  \sum_{w = 2}^{k+1} a^{(k+1)}_{\ell,w}  \gamma w \left(\hQ_{k, \lambda, \gamma} (u) \right)_{(k+1-w, w-1),(s, i)}.
\end{aligned}
\end{equation}

Let us turn to proving that Equation (\ref{eq-hatq}) holds true for $k+1$. Let $\ell \in \{1,\dots,k+1\}$ and consider $(s,i) \in \Omega(k+1)$ such that $s \leq (k+1)-\ell$ and define $m = (k+1) - (s+i)$. 

Notice that $m = 0$ if and only if $(s,i) \in \widehat{\Omega}(k+1)$ and $s = k+1-i$. As a consequence, if $i < \ell$, the set $\mathcal{I}_{k+1}{(\ell,0,i)}$ is empty; otherwise, if $i \geq \ell$, $\mathcal{I}_{k+1}{(\ell,0,i)} = \{(\ell,i)\}$. In both cases, the right hand side of Equation (\ref{eq-hatq}) equals $a^{(k+1)}_{\ell,i}$.
Thus, by Equation (\ref{eq-inverse-q}), Assertion (\ref{eq-hatq}) follows.

Consider now the case $m > 0$, \emph{i.e.} $(s,i) \in \Omega(k)$. It follows from Equation (\ref{eq-inverse-q}), using the change of variable $\ell' = w-1$, that 
\begin{equation}
\label{eq-part1}
\left(\hQ_{k+1, \lambda, \gamma} (u) \right)_{(k+1-\ell,\ell),(s, i)} = \sum_{\ell' = 1}^k a^{(k+1)}_{\ell, \ell' + 1} \gamma(\ell'+1) \left(\hQ_{k, \lambda, \gamma} (u) \right)_{(k-\ell',\ell'),(s,i)}.
\end{equation}

\noindent Using the inductive hypothesis and noticing that $k - (s+i) = m-1$, we get that for any $\ell' \in \{1, \dots k\}$,
\begin{equation*}
\begin{aligned}
\left(\hQ_{k, \lambda, \gamma} (u) \right)&_{(k-\ell',\ell'),(s,i)} =\\ 
& \frac{1}{u + \lambda si + \gamma i} \sum_{\mathfrak{i} \in \mathcal{I}_k(\ell',m-1,i)} \prod_{j=1}^m q_{k,\lambda,\gamma}(\mathfrak{i},j-1;u) g_{k,m-1,\lambda,\gamma}(\mathfrak{i},j-1;u).
\end{aligned}
\end{equation*}
Notice that if $\mathfrak{i} \in \N^{m+2}$ is such that $\mathfrak{i}_{j-1} = \mathfrak{i}'_j$ for all $j$, then $q_{k,\lambda,\gamma}(\mathfrak{i},j-1;u) = q_{k+1,\lambda,\gamma}(\mathfrak{i}',j;u)$ and $g_{k, m-1, \lambda,\gamma}(\mathfrak{i},j-1;u) = g_{k+1,m,\lambda,\gamma}(\mathfrak{i}',j;u)$.
Letting $\tau$ be the projection on $\N^{m+2}$:  for any $\mathfrak{i}' = (i_0,\dots,i_{m+1}) \in \N^{m+2}$, $\tau(\mathfrak{i}') = (i_1,\dots,i_{m+1})$, we get that:
\begin{equation*}
\begin{aligned}
\sum_{\mathfrak{i} \in \mathcal{I}_k(\ell',m-1,i)} \prod_{j=1}^m & q_{k,\lambda,\gamma}(\mathfrak{i},j-1;u) g_{k,m-1,\lambda,\gamma}(\mathfrak{i},j-1;u) = \\
& \sum_{\substack{\mathfrak{i}' \in \N^{m+2}: \\ \tau(\mathfrak{i}') \in \mathcal{I}_k(\ell',m-1,i)}} \prod_{j=1}^m q_{k+1,\lambda,\gamma}(\mathfrak{i}',j;u) g_{k+1,m,\lambda,\gamma}(\mathfrak{i}',j;u).
\end{aligned}
\end{equation*}

Furthermore, for $\mathfrak{i}' \in \N^{m+2}$ such that $i_0 = \ell, i_1 = \ell'$, it holds that 
\begin{equation}
\label{eq-part2}
a^{(k+1)}_{\ell, \ell' + 1} \gamma(\ell'+1) = \ind_{\{\ell' + 1 \geq \ell\}} q_{k+1,\lambda,\gamma}(\mathfrak{i}',0;u) g_{k+1,m,\lambda,\gamma}(\mathfrak{i}',0;u).
\end{equation}

Thus, noticing that the limits of the sum over $\ell'$ in Equation (\ref{eq-part1}) taken together with $\ind_{\{\ell' + 1 \geq \ell\}}$ from Equation (\ref{eq-part2}) induce that $\ell-1 \leq \ell' \leq k$, Equation (\ref{eq-part1}) yields the desired result:
\begin{equation*}
\begin{aligned}
\left(\hQ_{k+1, \lambda, \gamma} (u) \right)&_{(k+1-\ell,\ell),(s, i)} = \\
& \frac{1}{u + \lambda si + \gamma i} \sum_{\mathfrak{i}' \in \mathcal{I}_{k+1}(\ell,m,i)} \prod_{j=0}^m q_{k+1, \lambda,\gamma}(\mathfrak{i}',j;u) g_{k+1,m, \lambda,\gamma}(\mathfrak{i}',j;u).
\end{aligned}
\end{equation*}
This completes the proof.
\end{proof}

\section{Computation of the exponential growth rate for the \emph{SEIR} model with two levels of mixing}
\label{sec-SEIR-growthrate}

In the \emph{SEIR} model, subsequently to an infectious contact between an infected and a susceptible individual, the susceptible first becomes \emph{exposed} (\emph{E} state, assimilated to an infected but not yet infectious state) for a duration distributed according to an exponential law of parameter $\mu$, before entering the infectious state. Thus, the computation of the exponential growth rate as proposed by \cite{pellisEpidemicGrowthRate2011} needs to be adapted. 

For a population of size $k$, consider the Markov chain giving the numbers $(S_t,E_t,I_t)_{t \geq 0}$ of susceptible, exposed and infected individuals in the population at time $t \geq 0$ after the beginning of the epidemic, of transition rates
\begin{equation}
    \begin{aligned}
    \text{Transition} &\quad & \text{Rate} \\
    (s, e, i) \to (s - 1, e + 1, i) & \quad & \lambda s i; \\
    (s, e, i) \to (s, e-1, i+1) & \quad & \mu e;\\
    (s, e, i) \to (s, e, i-1) & \quad & \gamma i.
\end{aligned}
\end{equation}
The set of transient states is then given by 
\begin{equation*}
    \Omega'(k) = \{(s,e,i) \in (\N \cup 0)^3: s + e + i \leq k, e + i \geq 1\},
\end{equation*}
of cardinal $d'(k) = k(k+1)(k+5)/6$.
The restriction of the generator of this Markov chain to $\Omega(k)$ is given by $Q'_{\lambda, \gamma, \mu}(k)$ defined by: $\forall (s,e,i), (s',e',i') \in \Omega'(k)$,

\begin{equation}
\left(Q'_{\lambda, \gamma, \mu}(k)\right)_{(s,e,i),(s',e',i')} =
\begin{cases}
\begin{aligned}
- \lambda s i -\mu e - \gamma i \;\; & \text{ if } (s',e',i') = (s,e,i); \\
\lambda s i \;\; & \text{ if } (s',e',i') = (s-1,e+1, i); \\
\mu e \;\; & \text{ if } (s',e',i') = (s,e-1, i+1); \\
\gamma i \;\; & \text{ if } (s',e',i') = (s, e, i-1); \\
0 \;\; & \text { otherwise.}
\end{aligned}
\end{cases}
\end{equation}

Following the work of \cite{pellisEpidemicGrowthRate2011} and adopting the notations introduced in Section \ref{reductionODE}, one can then easily see that the exponential growth rate $r'$ of the \emph{SEIR} model with two levels of mixing is characterized by the implicit equation 
\begin{equation}
    \rho(K'(r')) = 1,
\end{equation}
where for $u \geq 0$, $K'(u)$ is the following matrix:
{\small
\begin{equation}
\label{eq-Kr-SEIR}
\begin{pmatrix} 
 \hat{\pi}^H \left(\LapG{\bullet}{\lambda_H, \mu}{\gamma}{\beta_G}'\right) (u) & \hat{\pi}^H \left(\LapG{\bullet}{\lambda_H, \mu}{\gamma}{\beta_G}'\right) (u) \\
\hat{\pi}^H \left(\Lap{\bullet}{\lambda_H, \mu}{\gamma}'\right) (u)  \hat{\pi}^W \left(\Lap{\bullet}{\lambda_W, \mu}{\gamma}'\right) (u) &  \left(1 + \hat{\pi}^H \left(\Lap{\bullet}{\lambda_H, \mu}{\gamma}'\right) (u)\right)  \hat{\pi}^W \left(\Lap{\bullet}{\lambda_W, \mu}{\gamma}'\right) (u)
\end{pmatrix},
\end{equation}
}
given the definition for $\beta, \lambda, \mu, \gamma > 0$ of
\begin{equation}
\label{eq:lapzeta_seir}
\begin{aligned}
\LapG{k}{\lambda, \mu}{\gamma}{\beta}'(u) &= \sum_{(s,e,i) \in \Omega'(k)} \beta i \left((uI_{d'(k)} - Q'_{\lambda, \gamma, \mu}(k))^{-1}\right)_{(k - 1,1,0),(s,e,i)}, \\
\Lap{k}{\lambda, \mu}{\gamma}'(u) &= \sum_{(s,e,i) \in \Omega'(k)} (\lambda si + \mu e) \left((uI_{d'(k)} - Q'_{\lambda, \gamma, \mu}(k))^{-1}\right)_{(k - 1,1,0),(s,e,i)}.
\end{aligned}
\end{equation}

\section{Numerical aspects for the model reductions of Section \ref{robust:reduction}}
\label{apx:seir-num}

\textbf{\emph{SIR} model with two levels of mixing and Gamma distributed individual recovery times}. The reduced model still takes the form of the dynamical system given in Equation \eqref{eq:reduction}, whose parameters are determined as follows. $\gamma$ is set to the average recovery time and the growth rate $r$ is obtained from Equation \eqref{eq-r}, where the coefficients of the relevant matrix defined in Equation \eqref{eq-Kr} are estimated through simulations.

\textbf{\emph{SEIR} model with two levels of mixing.}
The reduced model is a standard uniformly mixing deterministic \emph{SEIR} model, with infectious contact rate $\widehat{\lambda}$ and transition rate $\mu$ for the transition from \emph{E} to \emph{I}. As before, $\gamma$ designates the recovery rate of infected individuals.

Notice that, given the parameters $\widehat{\lambda}$, $\mu$ and $\gamma$, the epidemic growth rate of the deterministic \emph{SEIR} model can be computed by solving Equation \eqref{growth-rate-hSIR} with $\zeta(\tau) = (\widehat{\lambda}/\gamma) \omega(\tau)$, where $\omega(\tau)$ is the distribution of infection times for an infected individual in a uniformly mixed population and $\widehat{\lambda}/\gamma$ is the average  number of infections caused by an individual in early stages of the \emph{SEIR} epidemic.

In order to choose the value of $\widehat{\lambda}$ for the model reduction, we proceed as follows. First, the epidemic growth rate $r$ of the \emph{SEIR} model with two layers of mixing is computed following section \ref{sec-SEIR-growthrate}. It then remains to set $\widehat{\lambda}$ in such a way that the epidemic growth rate of the deterministic \emph{SEIR} model, which can be computed as described above, is equal to $r$. This can be achieved through simulations. Alternatively, it is possible to use the following formula derived from the Supplementary Material of \cite{trapmanInferringR0Emerging2016} which states that 
\begin{equation*}
    \widehat{\lambda} = \gamma \left(1 + \frac{r}{\gamma} \right) \left(1 = \frac{r}{\mu} \right).
\end{equation*}

The reduced model is then defined by the following set of ordinary differential equations:
\begin{equation}
\label{eq:reduction_seir}
\left\{
\begin{aligned}
\frac{dS}{dt} = & - \widehat{\lambda} S I \\
\frac{dE}{dt} = & \widehat{\lambda} S I - \mu E \\
\frac{dI}{dt} = & \mu E - \gamma I \\
\frac{dR}{dt} = & \gamma I \\
\end{aligned}
\right.
\end{equation}

\paragraph{Acknowledgements} The authors thank the reviewers for their pertinent and detailed feedback on the manuscript. This work has been supported by the Chair "Mod\'elisation Math\'ematique et Biodiversit\'e" of Veolia Environnement - Ecole Polytechnique - Museum national d'Histoire naturelle - Fondation X. \\
This work has also been partially funded by the European Union (ERC, SINGER, 101054787). Views and opinions expressed are however those of the author(s) only and do not necessarily reflect those of the European Union or the European Research Council. Neither the European Union nor the granting authority can be held responsible for them.

\end{appendices}

\bibliographystyle{apalike} 
\bibliography{epidemiological-footprint}

\newpage
\section*{Online resource - supplementary material}

\renewcommand{\thefigure}{S\arabic{figure}}

\begin{figure}[!ht]
  \centering
\includegraphics[width=0.7\textwidth]{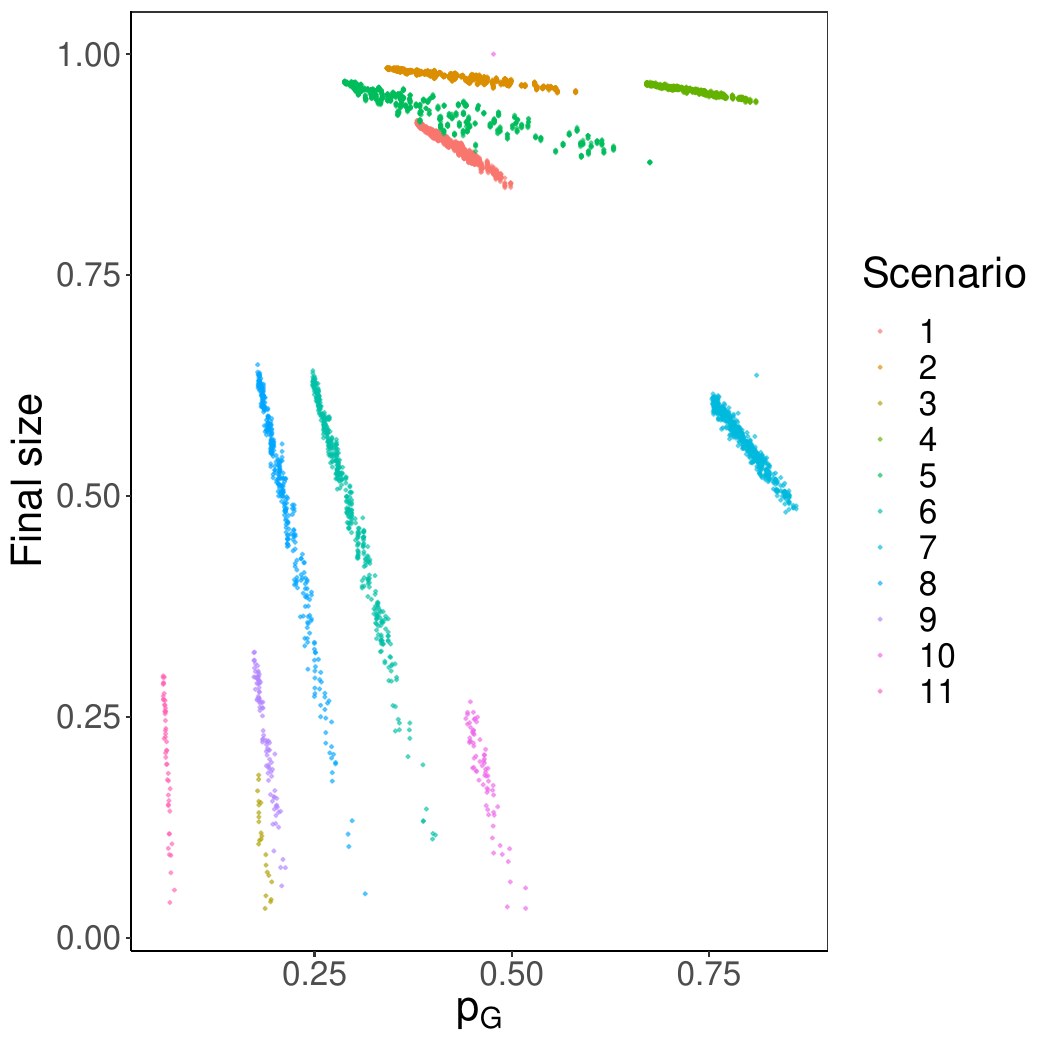} 
\caption{Simulated values of the final size as a function of the proportion of infection via global mixing ($p_G$) for all scenarios of Table 3 and all exploratory workplace size distributions from set A, Table 2. For each epidemic scenario and each workplace size distribution, simulations were repeated 10 times. Each coloured scatter plot on the figure thus contains 160x10 points. The household size distribution is the reference household size distribution.}
\label{fig_supp:scenario_size_pG}
\end{figure}

\begin{figure}[!ht]
  \centering
\includegraphics[width=0.7\textwidth]{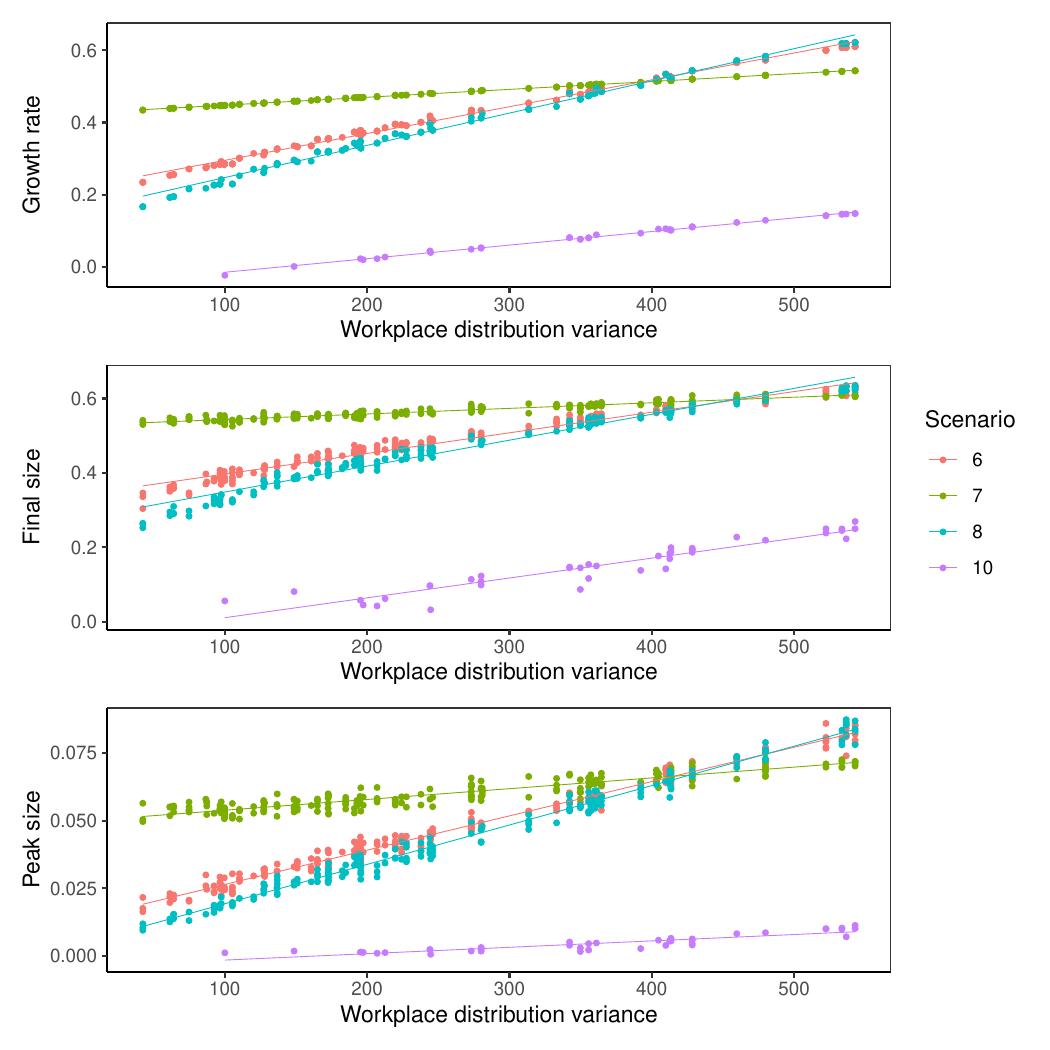} 
\caption{Influence of the variance of the workplace size distribution on the epidemic growth rate (top), final size (middle) and peak size (bottom). Simulations of the stochastic structured model with sub-linear infection rates in households and workplaces were performed with the reference household size distribution, exploratory workplace size distribution set B with average workplace size of 20 from Table 2 and epidemic scenarios 6, 7, 8 and 10 from Table 3. Simulations were repeated 10 times for each combination of scenario and workplace size distribution.
}
\label{fig_supp:Variance2}
\end{figure}

\begin{figure}[!ht]
  \centering
\includegraphics[width=0.7\textwidth]{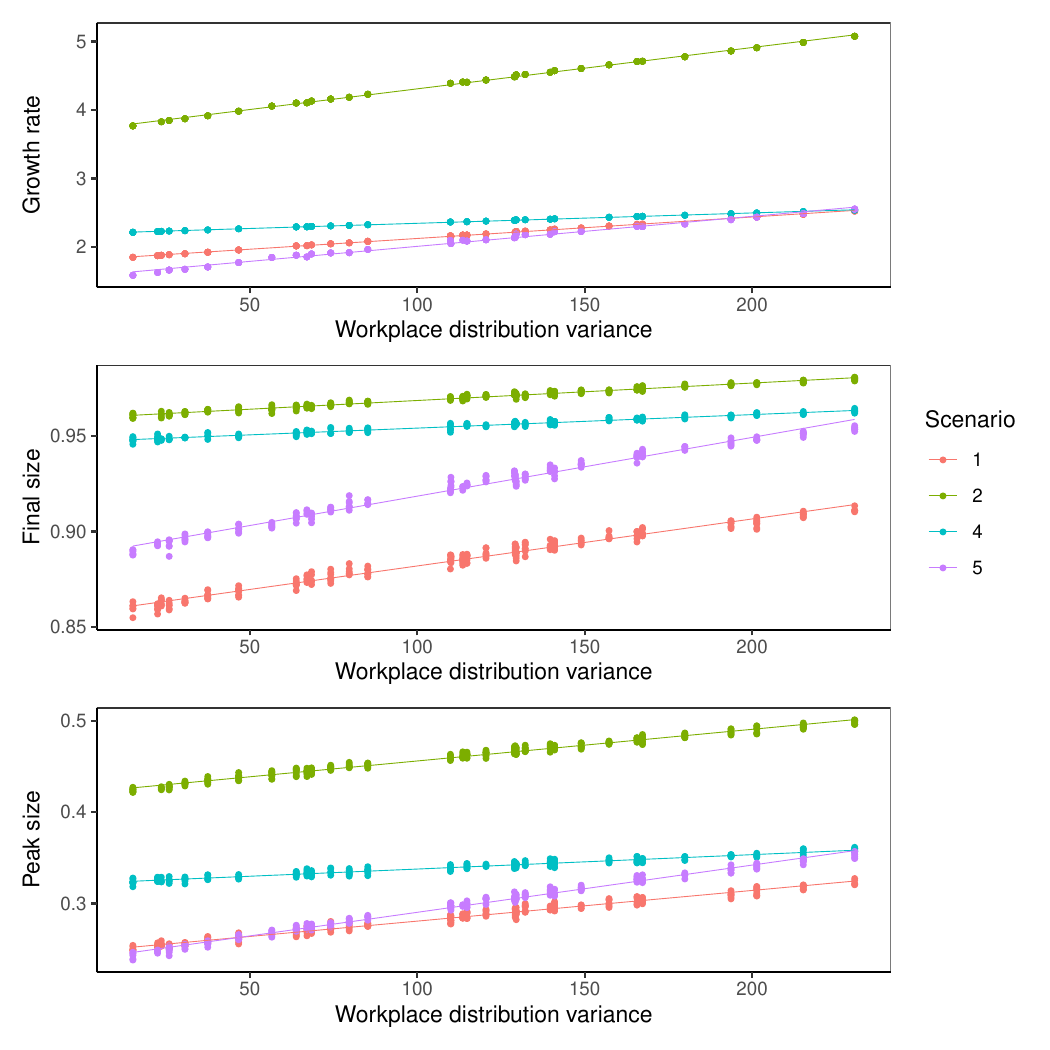} 
\caption{Influence of the variance of the workplace size distribution on the epidemic growth rate (top), final size (middle) and peak size (bottom). Simulations of the stochastic structured model with linear infection rates in households and workplaces were performed with the reference household size distribution, exploratory workplace size distribution set C with average workplace size of 7 from Table 2 and epidemic scenarios 1, 2, 4 and 5 from Table 3. Simulations were repeated 10 times for each combination of scenario and workplace size distribution.
}
\label{fig_supp:Variance_mean7_1245}
\end{figure}

\begin{figure}[!ht]
  \centering
\includegraphics[width=0.7\textwidth]{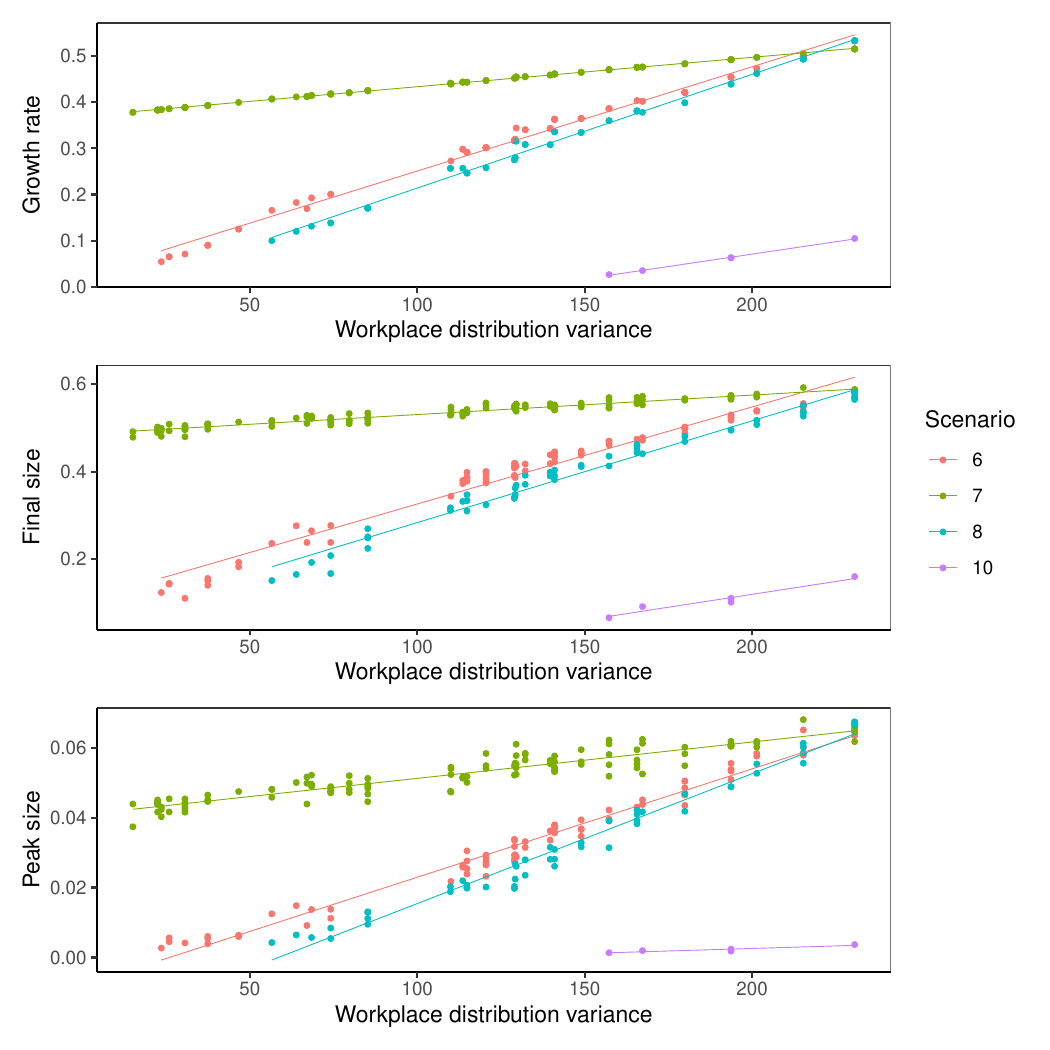} 
\caption{Influence of the variance of the workplace size distribution on the epidemic growth rate (top), final size (middle) and peak size (bottom). Simulations of the stochastic structured model with linear infection rates in households and workplaces were performed with the reference household size distribution, exploratory workplace size distribution set C with average workplace size of 7 from Table 2 and epidemic scenarios 6, 7, 8 and 10 from Table 3. Simulations were repeated 10 times for each combination of scenario and workplace size distribution.
}
\label{fig_supp:Variance_mean7_67810}
\end{figure}

\begin{figure}[!ht]
  \centering
\includegraphics[width=0.7\textwidth]{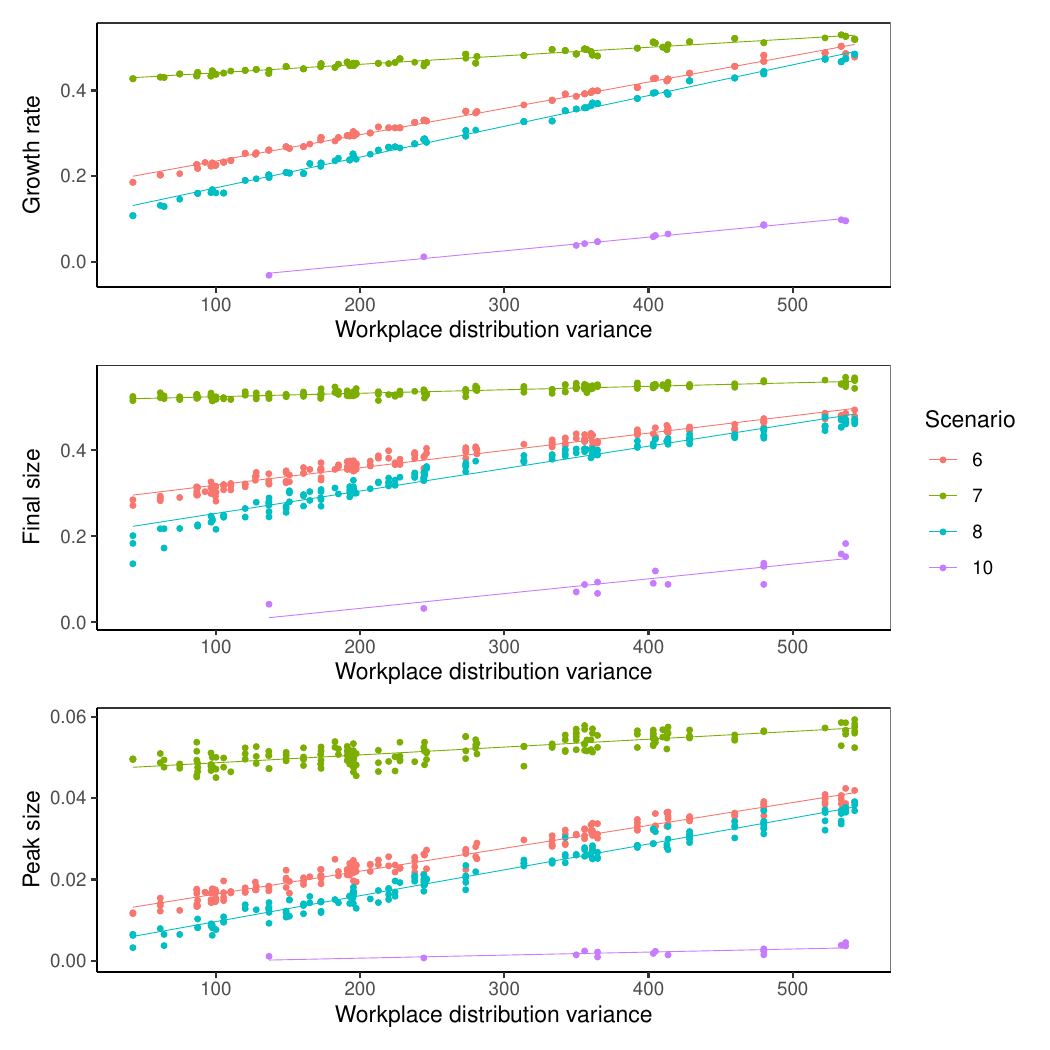} 
\caption{Influence of the variance of the workplace size distribution on the epidemic growth rate (top), final size (middle) and peak size (bottom). Simulations of the stochastic structured model with sub-linear infection rates in households and workplaces were performed with the reference household size distribution, exploratory workplace size distribution set B with average workplace size of 20 from Table 2 and epidemic scenarios 6, 7, 8 and 10 from Table 3. Simulations were repeated 10 times for each combination of scenario and workplace size distribution.
}
\label{fig_supp:Variance_sqrt2}
\end{figure}

\begin{figure}[!ht]
  \centering
\includegraphics[width=0.7\textwidth]{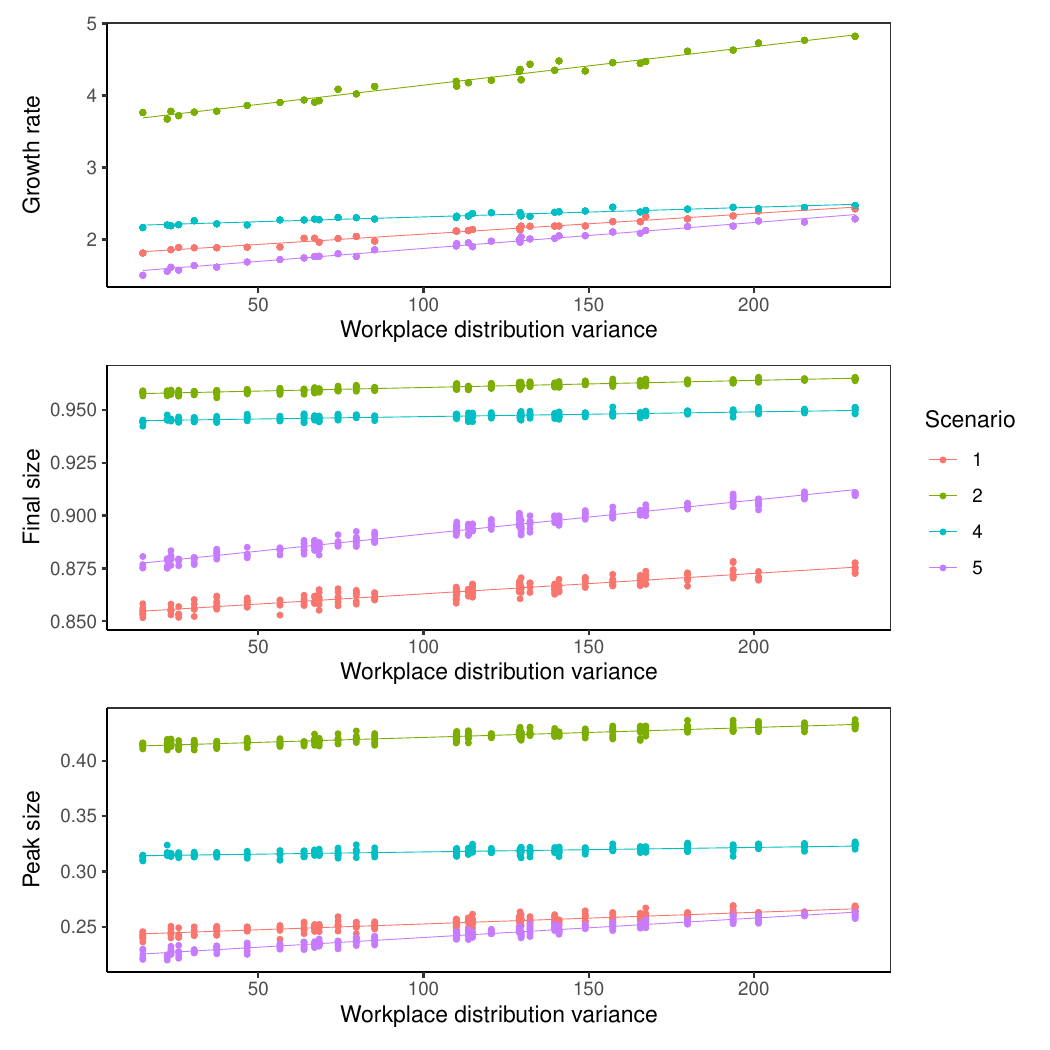} 
\caption{Influence of the variance of the workplace size distribution on the epidemic growth rate (top), final size (middle) and peak size (bottom). Simulations of the stochastic structured model with sub-linear infection rates in households and workplaces were performed with the reference household size distribution, exploratory workplace size distribution set C with average workplace size of 7 from Table 2 and epidemic scenarios 1, 2, 4 and 5 from Table 3. Simulations were repeated 10 times for each combination of scenario and workplace size distribution.
}
\label{fig_supp:Variance_mean7_1245_sqrt}
\end{figure}

\begin{figure}[!ht]
  \centering
\includegraphics[width=0.7\textwidth]{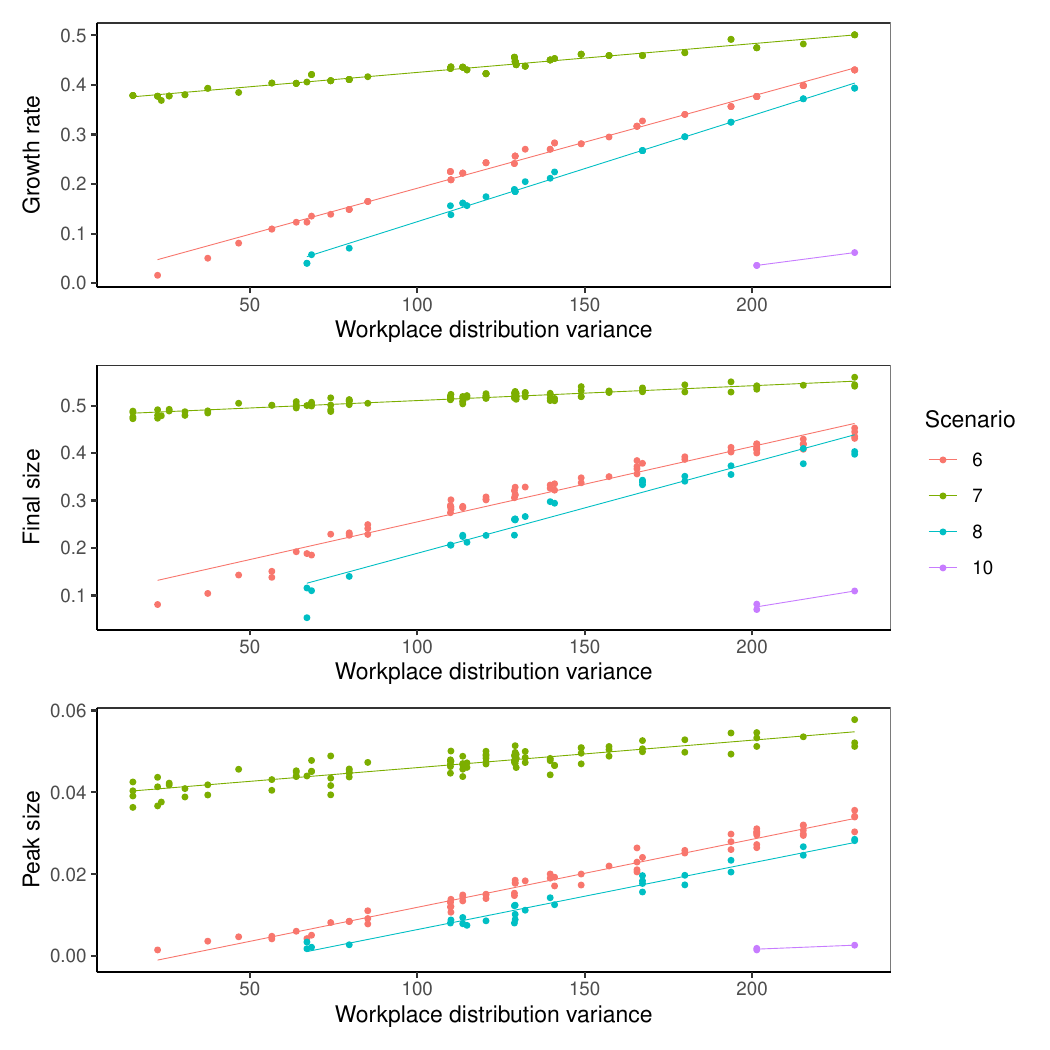} 
\caption{Influence of the variance of the workplace size distribution on the epidemic growth rate (top), final size (middle) and peak size (bottom). Simulations of the stochastic structured model with sub-linear infection rates in households and workplaces were performed with the reference household size distribution, exploratory workplace size distribution set C with average workplace size of 7 from Table 2 and epidemic scenarios 6, 7, 8 and 10 from Table 3. Simulations were repeated 10 times for each combination of scenario and workplace size distribution.
}
\label{fig_supp:Variance_mean7_67810_sqrt}
\end{figure}

\begin{figure}[!ht]
\centering
\includegraphics[width=0.7\textwidth]{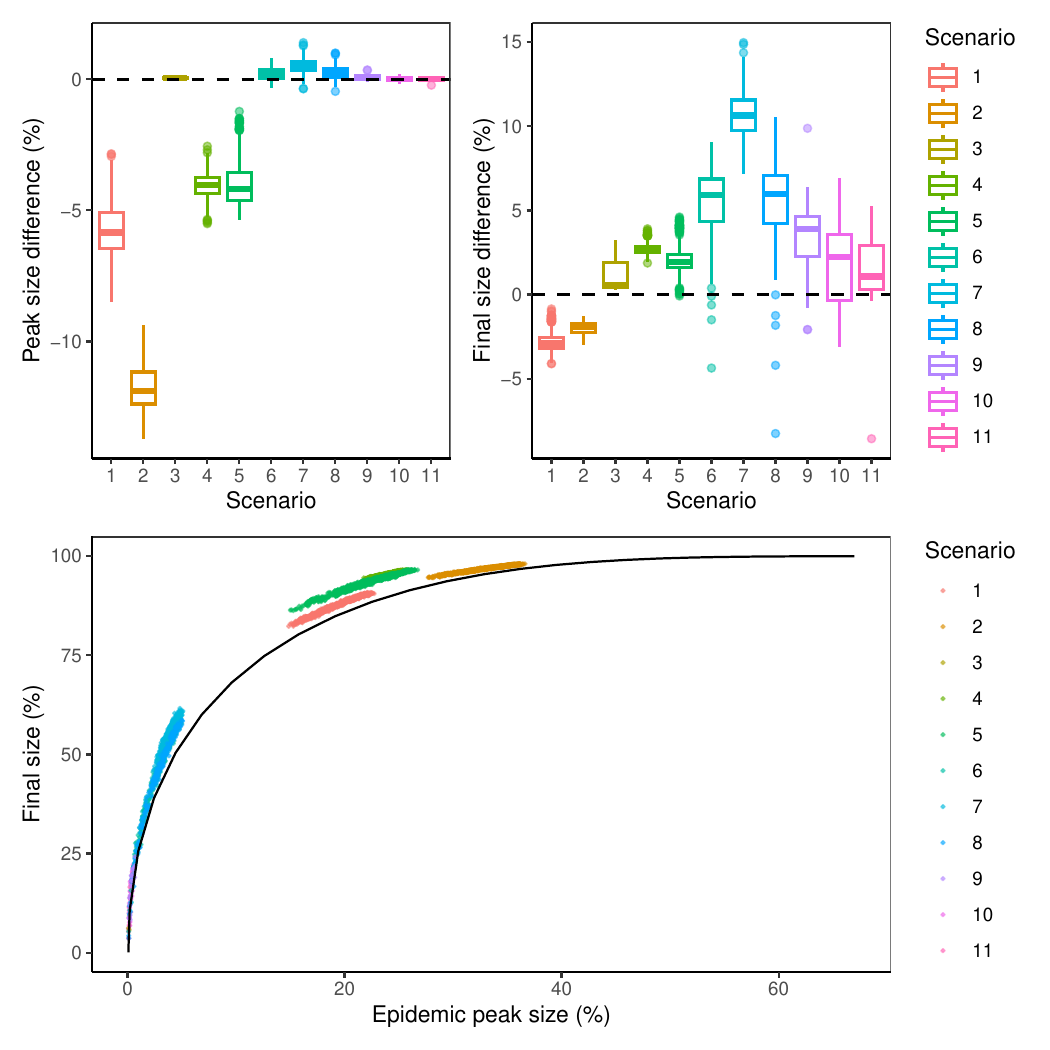}
\caption{
Boxplot of the differences in the peak sizes and (left column) and final sizes (right column) between the complete and the reduced models,
obtained from simulations of the stochastic structured \emph{SIR} model,
with gamma distributed recovery rates.
The shape parameter $a$ and rate parameter $r$ are set to values of $a=r=0.5$, \emph{i.e.} the density is given by $(r^{-a} \Gamma(a))^{-1} x^{a-1}e^{-rx}$. The average recovery rate is 1.
Simulations of the stochastic structured model are performed with population size of 100,000,
reference household distribution, workplace size distributions set A, for all scenarios from Table 3.
Each combination of scenario and workplace size distribution from set A is repeated 10 times.
Only simulations where an epidemic outbreak occurred (\emph{i.e.} more than 3\% of the population become infected) are reported in this figure.}
\label{fig_supp:sir_gamma_05_1}
\end{figure}

\begin{figure}[!ht]
\centering
\includegraphics[width=0.7\textwidth]{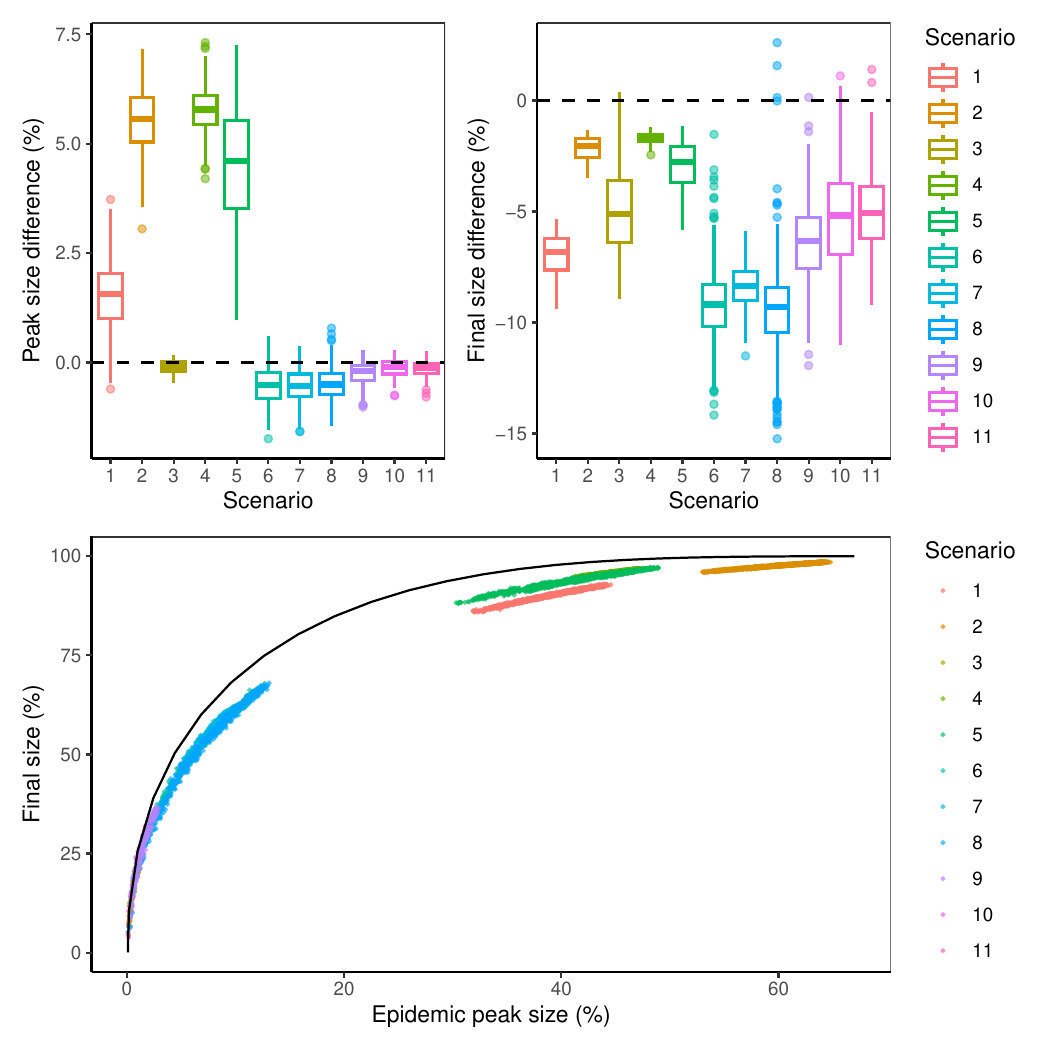}
\caption{
Boxplot of the differences in the peak sizes and (left column) and final sizes (right column) between the complete and the reduced models,
obtained from simulations of the stochastic structured \emph{SIR} model,
with gamma distributed recovery rates.
The shape parameter $a$ and rate parameter $r$ are set to values of $a=r=2$, \emph{i.e.} the density is given by $(r^{-a} \Gamma(a))^{-1} x^{a-1}e^{-rx}$. The average recovery rate is 1.
Simulations of the stochastic structured model are performed with population size of 100,000,
reference household distribution, workplace size distributions set A, for all scenarios from Table 3.
Each combination of scenario and workplace size distribution from set A is repeated 10 times.
Only simulations where an epidemic outbreak occurred (\emph{i.e.} more than 3\% of the population become infected) are reported in this figure.}
\label{fig_supp:sir_gamma_05_2}
\end{figure}

\begin{figure}[!ht]
  \centering
\includegraphics[width=0.7\textwidth]{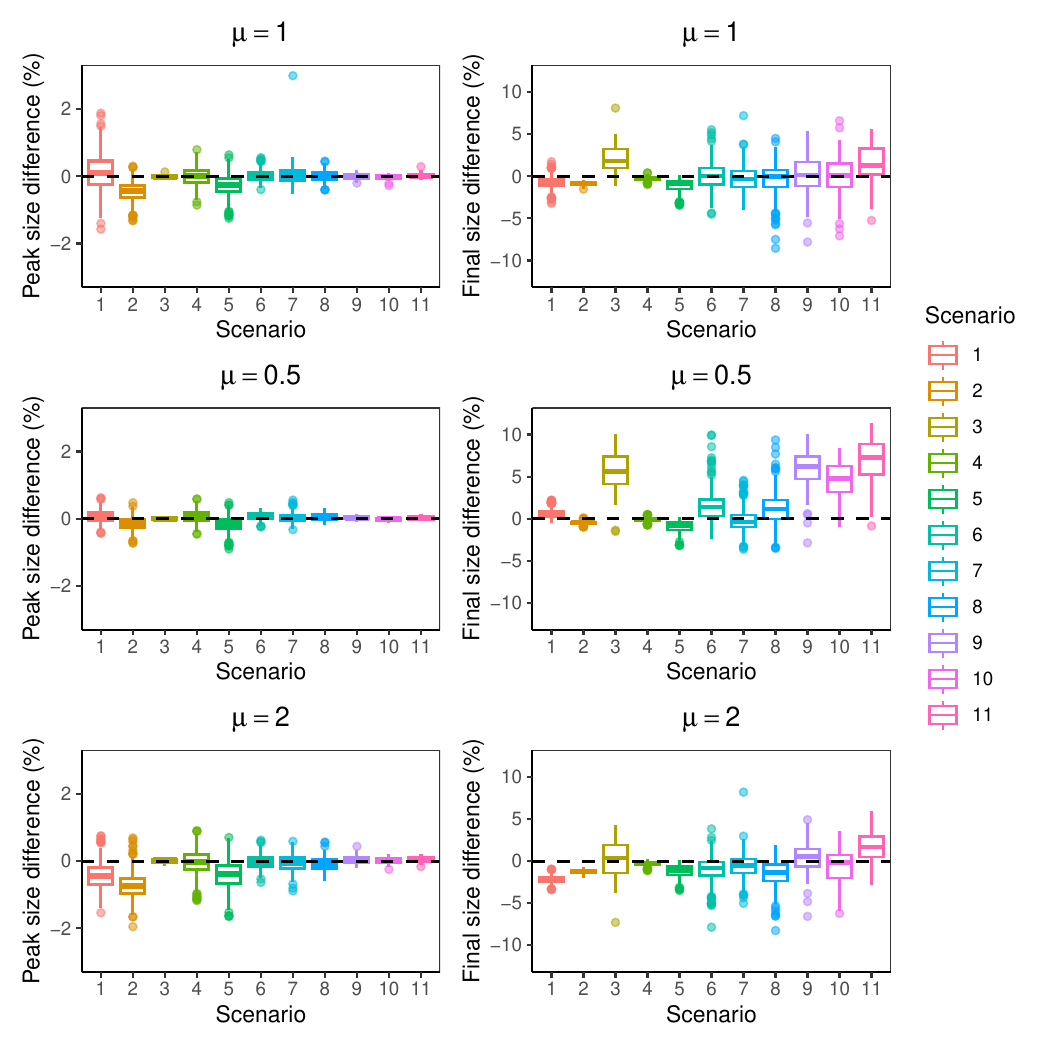} 
\caption{
Boxplot of the differences in the peak sizes (left column) and final sizes (right column) between the complete and the reduced models,
obtained from simulations of the stochastic structured \emph{SEIR} model and the corresponding reduced model defined in Equation (F16),
for values of the transition rate from \emph{E} to \emph{I}, $\mu=1$, $\mu=0.5$ and $\mu=2$.
Simulations of the stochastic structured model are performed with population size of 100,000,
reference household distribution, workplace size distributions set A, for all scenarios from Table 3.
Each combination of scenario and workplace size distribution from set A is repeated 10 times.
Only simulations where an epidemic outbreak occurred (\emph{i.e.} more than 3\% of the population become infected) are reported in this figure.
}
\label{fig_supp:reduction_robustnessSEIR_1}
\end{figure}

\begin{figure}[!ht]
  \centering
\includegraphics[width=0.7\textwidth]{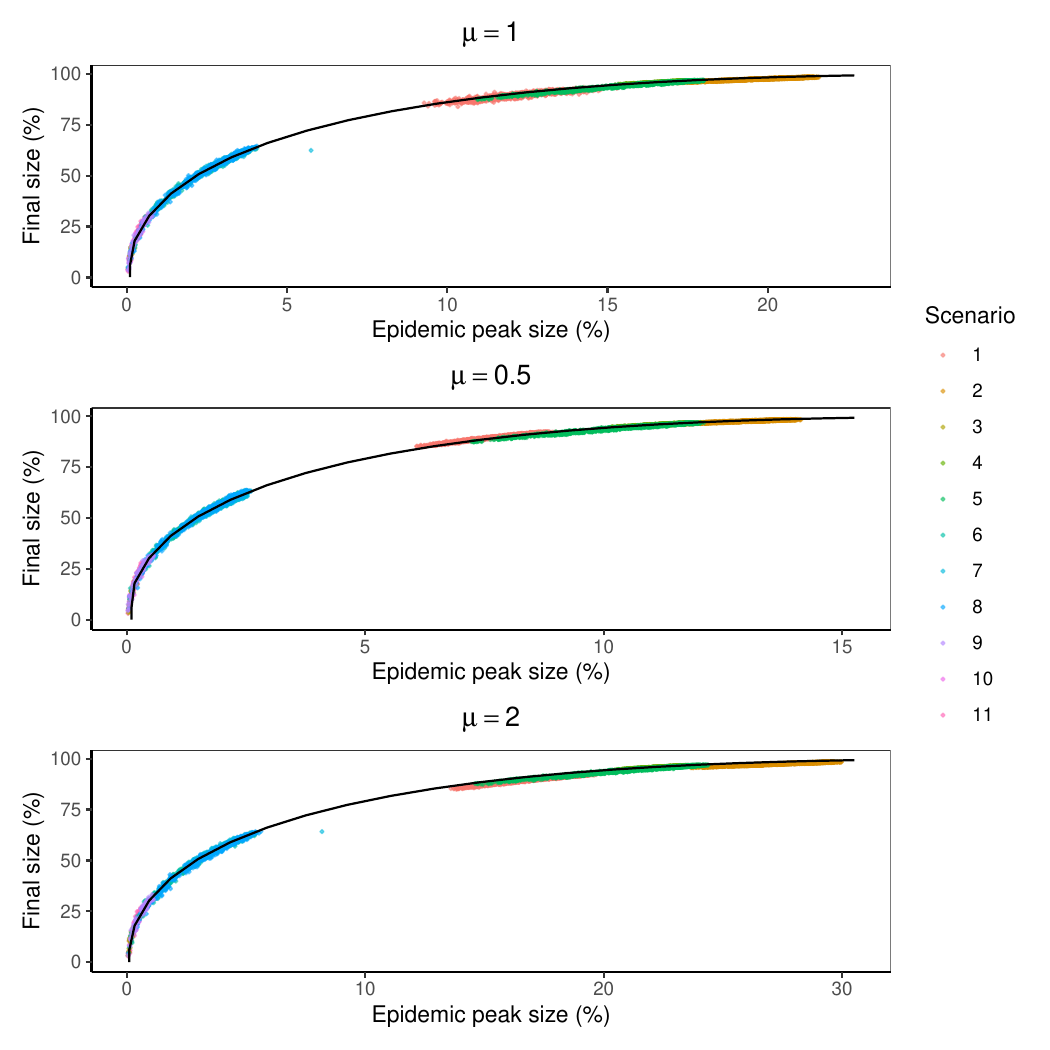} 
\caption{
Plot of the epidemic final size as a function of the  epidemic peak sizes 
obtained from simulations of the stochastic structured \emph{SEIR} model,
for values of the transition rate from \emph{E} to \emph{I}, $\mu=1$, $\mu=0.5$ and $\mu=2$.
The black line represents the epidemic peak size and epidemic final size for the standard \emph{SEIR} model.
Simulations of the stochastic structured model are performed with population size of 100,000,
reference household distribution, workplace size distributions set A, for all scenarios from Table 3.
Each combination of scenario and workplace size distribution from set A is repeated 10 times.
Only simulations where an epidemic outbreak occurred (\emph{i.e.} more than 3\% of the population become infected) are reported in this figure.
}
\label{fig_supp:reduction_robustnessSEIR_2}
\end{figure}

\begin{figure}[!ht]
  \centering
\includegraphics[width=0.7\textwidth]{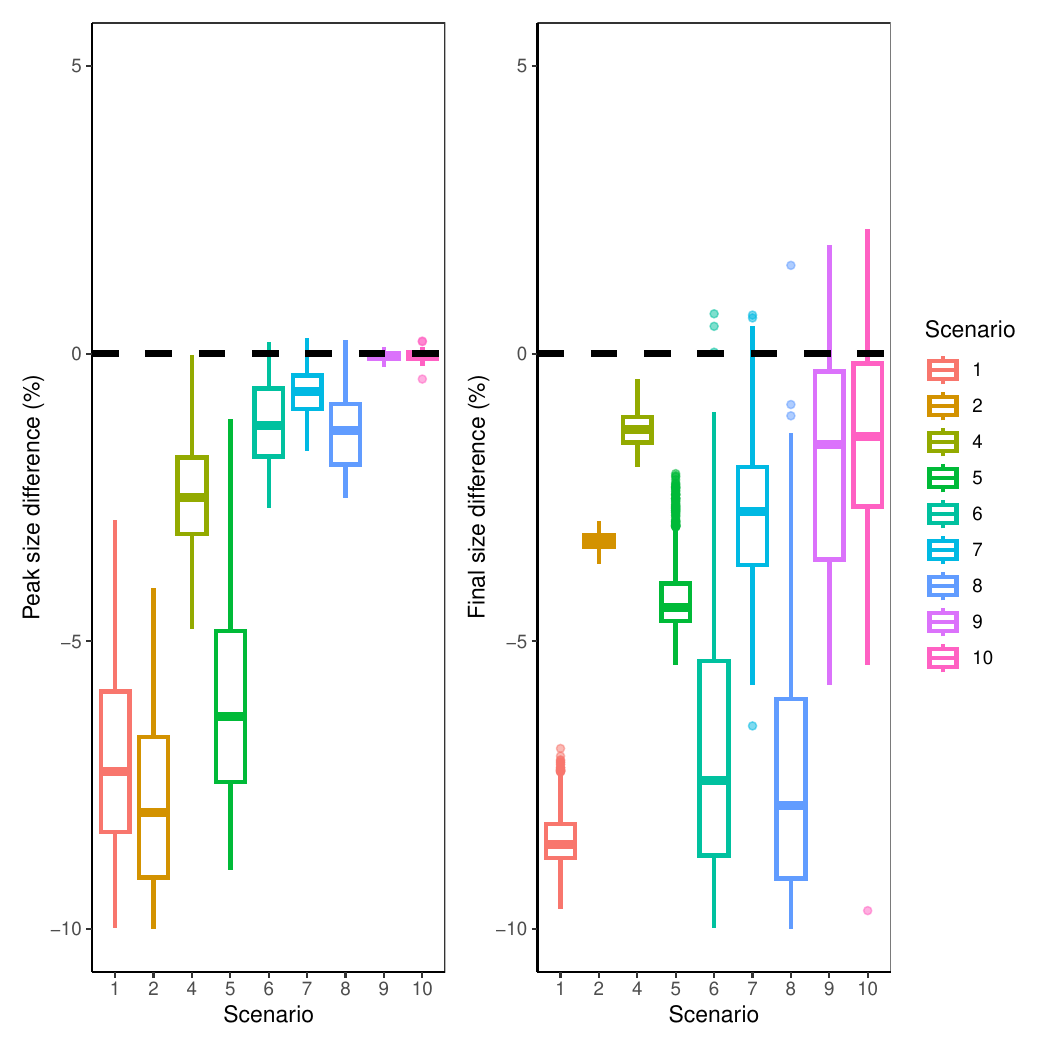} 
\caption{
Boxplot of the differences in the peak sizes and (left) and final sizes (right) between the complete and the reduced models, obtained from simulations of the stochastic structured model with sub-linear rate and the corresponding reduced model. Simulations of the stochastic structured model are performed with population size of 100,000, reference household distribution, workplace size distributions set A, for all scenarios from Table 3 (point color). Each combination of scenario and workplace size distribution from set A is repeated 10 times. Only simulations where an epidemic outbreak occurred (\emph{i.e.} more than 3\% of the population become infected) are reported in this figure.}
\label{fig_supp:reduction_robustness2_sqrt}
\end{figure}

\end{document}